\documentclass[english]{article}
\usepackage[T1]{fontenc}
\usepackage[latin1]{inputenc}
\usepackage{geometry}
\usepackage{float}
\usepackage{mathrsfs}
\usepackage{amsmath}
\usepackage{amssymb}
\usepackage{graphicx}
\usepackage{subfig}
\usepackage{colortbl} 
\usepackage{hhline} 

\definecolor{Gray}{gray}{0.85}
\makeatletter
\floatstyle{ruled}
\newfloat{algorithm}{tbp}{loa}
\providecommand{\algorithmname}{Algorithm}
\floatname{algorithm}{\protect\algorithmname}

\usepackage{amsthm}

\usepackage{mathrsfs}

\usepackage{amsfonts}

\usepackage{epsfig}

\usepackage{bm}

\usepackage{mathrsfs}

\usepackage{enumerate}

\@ifundefined{definecolor}{\@ifundefined{definecolor}
 {\@ifundefined{definecolor}
 {\usepackage{color}}{}
}{}
}{}

\usepackage{hyperref}
\hypersetup{
    colorlinks=true,
    linkcolor=blue,
    filecolor=magenta,      
    urlcolor=cyan,
    citecolor = red,
}

\usepackage{subfig}\usepackage[all]{xy}

\newtheorem{lem}{Lemma}[section]
\newtheorem{rem}{Remark}[section]
\newtheorem{prop}{Proposition}[section]

\newcounter{hypA}

\newcounter{hypB}

\newcounter{hypD}
\newenvironment{hypD}{\refstepcounter{hypD}\begin{itemize}
 \item[({\bf D\arabic{hypD}})]}{\end{itemize}}

\usepackage{babel}\date{}

\usepackage{babel}

\makeatother

\usepackage{babel}

\begin{document}

\begin{center}

{\Large \textbf{Unbiased and Multilevel Methods for a Class of Diffusions Partially Observed via Marked Point Processes}}

\vspace{0.5cm}

BY MIGUEL ALVAREZ, AJAY JASRA \& HAMZA RUZAYQAT

{\footnotesize Applied Mathematics and Computational Science Program, 
Computer, Electrical and Mathematical Sciences and Engineering Division, King Abdullah University of Science and Technology, Thuwal, 23955, KSA.}
{\footnotesize E-Mail:\,} \texttt{\emph{\footnotesize miguelangel.alvarezballesteros@kaust.edu.sa, ajay.jasra@kaust.edu.sa, hamza.ruzayqat@kaust.edu.sa}}
\end{center}

\begin{abstract}
In this article we consider the filtering problem associated to partially observed diffusions, with observations following a marked point process.
In the model, the data form a point process with observation times that have its intensity driven by a diffusion, with the associated marks also depending upon the diffusion process.
We assume that one must resort to time-discretizing the diffusion process and develop particle and multilevel particle filters to recursively approximate the filter.
In particular, we prove that our multilevel particle filter can achieve a mean square error (MSE) of $\mathcal{O}(\epsilon^2)$ ($\epsilon>0$ and arbitrary) with a cost
of $\mathcal{O}(\epsilon^{-2.5})$ versus using a particle filter which has a cost of $\mathcal{O}(\epsilon^{-3})$ to achieve the same MSE. We then show how this methodology can be extended to give unbiased (that is with no time-discretization error) estimators of the filter, which are proved to have finite variance and with high-probability have finite cost. Finally, we extend our methodology to the problem of online static-parameter estimation.
\\ 
\noindent \textbf{MSC classes}:	60G55, 60G35, 62M20, 62F30

\noindent\textbf{Code available at:} \href{https://github.com/maabs/Multilevel-for-Diffusions-Observed-via-Marked-Point-Processes}{https://github.com/maabs/Multilevel-for-Diffusions-Observed-via-Marked-Point-Processes}\\
\noindent\textbf{Corresponding author}: Miguel Alvarez. E-mail:
\href{mailto:miguelangel.alvarezballesteros@kaust.edu.sa}{miguelangel.alvarezballesteros@kaust.edu.sa}
\end{abstract}

\section{Introduction}

The class of partially observed diffusion processes can be found in wide variety of real applications; see for instance the coverage in \cite{crisan_bain,delm:2004}. In this paper we consider observation models which constitute marked point processes and in particular, when the points (observation times) are a non-homogeneous Poisson process with intensity dependent upon the diffusion; this has applications for instance in earthquake modeling see \cite{fearn,jin,martin}. A related problem is that of trying to also, simultaneously, estimate the static (time-homogeneous) parameters of the model and this is also considered in this article.

The filtering of partially observed diffusion processes is notoriously challenging as even if the transition dynamics or density of the diffusion is available, one must resort to numerical approximation of the filter. Such approaches, at least for low-dimensional diffusions, are often focussed upon the particle filter methodology; see \cite{crisan_bain,delm:2004,fearn,mlpf} for example. In many cases, one is unable to exactly sample from the transition of the diffusion or its transition density is unavailable even up-to a non-negative and unbiased estimator (see \cite{fearn} for diffusions where this is possible) and so often one has to resort to using a time-discretization of the diffusion.

In the scenario where one considers approximating the time-discretized filter as well as estimating the static parameters, there have been several works including \cite{archibald, beskos,del_jacod,jin,mlpf,mlpf_nc,mlpf_new} to name but a few; as just noted these methods will have bias in terms of the time-discretization. All of these articles depend, in some way, on the particle filter method. This approach will generate a collection of samples in parallel which undergo transitions using the dynamics of the discretized diffusion and then are weighted and resampled to approximate the filter; see \cite{crisan_bain,delm:2004}. This has been further extended in the works of \cite{beskos,mlpf,mlpf_nc,mlpf_new} to incorporate the multilevel Monte Carlo (MLMC) method; see \cite{giles,giles1,hein} for the initial articles and \cite{ml_rev} for a review when using MLMC in data problems. This technique relies upon simulating approximations of the filters at multiple different levels of time-discretizations. The approach in  \cite{mlpf,mlpf_nc,mlpf_new} is to use a coupling of the time-discretized dynamics and of the resampling operation to give a type of coupled particle filter. It has been shown to reduce the computational cost to achieve a pre-specified mean square error (MSE) versus using an ordinary particle filter.
\cite{beskos} consider the problem of on-line parameter estimation for a type of partially observed diffusion process, using particle filters. 

To the best of our knowledge, none of the above papers consider the observation model that is under-study in this paper, with the exception of \cite{fearn,jin}. The work of \cite{fearn} focusses upon using the exact simulation of diffusions methodology to produce unbiased estimators of the filter. As is well-known, such methodology rarely works in dimensions bigger than 1 and often places significant restrictions on the diffusions that can be used, albeit being very elegant and useful in the scenarios where it can be adopted. The approach of this paper really only needs the strong error of the 
 method to fall at a fast enough rate; often this means only that the drift and diffusion coefficients need Lipschitz and growth conditions; see \cite{kloeden}. The method of \cite{jin}, which provides unbiased inference (removes time discretization bias), relies both on the exact simulation of the diffusion, which again is a restricted class of problem, for which there is a transition density and the well-known Poisson estimator \cite{wagner}, which in some examples can have a substantial variance. In addition, the authors often have to truncate their estimator to yield unbiased estimators as there is a (small) probability that the estimator of the exponential of the integrated intensity is negative. 
\cite{jin} also extend their ideas in the context of Markov chain Monte Carlo and batch (fixed data set) static Bayesian parameter estimation; the approach in this article can be extended to performing either a multilevel parameter estimation using Euler-Maruyama scheme as in \cite{jasra_bpe_sde}, or using Runge-Kutta schemes as in \cite{del_Bayes} or unbiased Bayesian parameter estimation for diffusions driven by a Brownian motion \cite{chada_ub} or diffusions driven by L\'{e}vy processes \cite{bayes_levy}, with the ideas of the afore-mentioned references, although we do not do so.
Whilst both of the methods in \cite{fearn,jin} clearly have several interesting contributions, they are not as general as the approach that we will develop.

In this article we apply and analyze methodology for  the filtering problem associated to partially observed diffusions, with observations following a marked point process. In particular,
we apply the multilevel particle filter (MLPF) \cite{mlpf} and unbiased particle filter (UPF) \cite{ub_pf} adapted to these models. The latter method is able to approximate the filter with no time-discretization bias, whilst only working with time-discretized dynamics. We also apply the ideas from \cite{beskos} to provide a method which can estimate the static (time-homogeneous) parameters of the model sequentially in time (online static parameter estimation) via maximum likelihood approach. To summarize, the contributions of this paper are as follows:
\begin{itemize}
\item{We develop the MLPF and UPF for the model under study.}
\item{We prove that the MLPF estimator can achieve a mean square error (MSE) of $\mathcal{O}(\epsilon^2)$, for arbitrary $\epsilon>0$, at a cost that is $\mathcal{O}(\epsilon^{-2.5})$.
If one uses a particle filter, then the cost is $\mathcal{O}(\epsilon^{-3})$ to achieve the same MSE.}
\item{The afore-mentioned technical results also allow us to show that our UPF estimator is unbiased, of finite variance and with high-probability (see e.g.~\cite{rhee}) has finite cost.}
\item{We adapt the approach in \cite{beskos} for online static parameter estimation.}
\item{We verify our findings by  implementing the methodology on several examples.}
\end{itemize}

This paper is structured as follows. In Section \ref{sec:model} we discuss the model and our associated algorithms for filtering, multilevel filtering and unbiased filtering as well as for online parameter estimation. In Section \ref{sec:math} we present our mathematical results. In Section \ref{sec:numerics} we provide several numerical examples that illustrate our algorithms and mathematical analysis. The proofs for our mathematical results can be found in the appendix.

\section{Modeling and Algorithms}\label{sec:model}

\subsection{Notations}\label{sec:notation}

Let $(\mathsf{X},\mathcal{X})$ be a measurable space.
For $\varphi:\mathsf{X}\rightarrow\mathbb{R}$ we write $\mathcal{B}_b(\mathsf{X})$ as the collection of bounded measurable functions. 
Let $\varphi:\mathbb{R}^d\rightarrow\mathbb{R}$, $\textrm{Lip}_{\|\cdot\|_2}(\mathbb{R}^{d})$ denotes the collection of real-valued functions that are Lipschitz with respect to (w.r.t.)~$\|\cdot\|_2$ ($\|\cdot\|_p$ denotes the $\mathbb{L}_p-$norm of a vector $x\in\mathbb{R}^d$). That is, $\varphi\in\textrm{Lip}_{\|\cdot\|_2}(\mathbb{R}^{d})$ if there exists a $C<+\infty$ such that for any $(x,y)\in\mathbb{R}^{2d}$
$$
|\varphi(x)-\varphi(y)| \leq C\|x-y\|_2.
$$
We write $\|\varphi\|_{\textrm{Lip}}$ as the Lipschitz constant of a function $\varphi\in\textrm{Lip}_{\|\cdot\|_2}(\mathbb{R}^{d})$.
For $\varphi\in\mathcal{B}_b(\mathsf{X})$, we write the supremum norm $\|\varphi\|=\sup_{x\in\mathsf{X}}|\varphi(x)|$.
$\mathcal{P}(\mathsf{X})$  denotes the collection of probability measures on $(\mathsf{X},\mathcal{X})$.
For a measure $\mu$ on $(\mathsf{X},\mathcal{X})$
and a function $\varphi\in\mathcal{B}_b(\mathsf{X})$, the notation $\mu(\varphi)=\int_{\mathsf{X}}\varphi(x)\mu(dx)$ is used. 
$B(\mathbb{R}^d)$ denote the Borel sets on $\mathbb{R}^d$. $dx$ is used to denote the Lebesgue measure.
If $K:\mathsf{X}\times\mathcal{X}\rightarrow[0,\infty)$ is a non-negative operator and $\mu$ is a measure, we use the notations
$
\mu K(dy) = \int_{\mathsf{X}}\mu(dx) K(x,dy)
$
and for $\varphi\in\mathcal{B}_b(\mathsf{X})$, 
$
K(\varphi)(x) = \int_{\mathsf{X}} \varphi(y) K(x,dy).
$
For $A\in\mathcal{X}$, the indicator function is written as $\mathbb{I}_A(x)$.
$\mathcal{N}_s(\mu,\Sigma)$ (resp.~$\psi_s(x;\mu,\Sigma)$)
denotes an $s-$dimensional Gaussian distribution (density evaluated at $x\in\mathbb{R}^s$) of mean $\mu$ and covariance $\Sigma$. If $s=1$ we omit the subscript $s$. For a vector/matrix $X$, $X^*$ is used to denote the transpose of $X$.
For $A\in\mathcal{X}$, $\delta_A(du)$ denotes the Dirac measure of $A$, and if $A=\{x\}$ with $x\in \mathsf{X}$, we write $\delta_x(du)$. 
For a vector-valued function in $d-$dimensions (resp.~$d-$dimensional vector), $\varphi(x)$ (resp.~$x$) say, we write the $i^{\textrm{th}}-$component ($i\in\{1,\dots,d\}$) as $\varphi^{(i)}(x)$ (resp.~$x^{(i)}$). For a $d\times q$ matrix $x$, we write the $(i,j)^{\textrm{th}}-$entry as $x^{(ij)}$.

\subsection{Model}\label{sec:mod}

We consider data $Y_{s_1},Y_{s_2},\dots$, $Y_{s_k}\in\mathsf{Y}\subseteq\mathbb{R}^{d_y}$, $k\in\mathbb{N}$, that are observed at the event times $s_1,s_2,\dots$, $s_k\in\mathbb{R}^+$ of a non-homogenous Poisson process driven by a diffusion. More precisely, consider a diffusion process
\begin{equation}\label{eq:diff}
dX_t = b(X_t)dt + \sigma(X_t)dW_t
\end{equation}
where $X_0=x_{*}\in\mathbb{R}^{d_x}$ given, $b:\mathbb{R}^{d_x}\rightarrow\mathbb{R}^{d_x}$, $\sigma:\mathbb{R}^{d_x}\rightarrow\mathbb{R}^{d_x\times d_x}$ is non-constant (except in some of our numerical examples) and $\{W_t\}_{t\geq 0}$ is a standard Brownian motion. At this stage we do not mention any static parameters; this is considered in Section \ref{sec:online}.
To minimize certain technical difficulties, the following assumption is made throughout the paper:
\begin{hypD}\label{hyp_diff:1}
We have:
\begin{enumerate}
\item{$\sigma^{(ij)}$ is bounded with $\sigma^{(ij)}\in\textrm{Lip}_{\|\cdot\|_2}(\mathbb{R}^{d_x})$, for all $(i,j)\in\{1,\dots,d_x\}^2$ and 
$$
a(x):=\sigma(x)\sigma(x)^*
$$ 
is uniformly elliptic for all $x\in\mathbb{R}^{d_x}$.}
\item{$b^{(j)}$ are bounded and $b^{(j)}\in\textrm{Lip}_{\|\cdot\|_2}(\mathbb{R}^{d_x})$, for all $j\in\{1,\dots,d_x\}$.}
\end{enumerate}
\end{hypD}
Note that the assumptions here and later are far from minimal. In general, as stated in the introduction, it should be enough that the drift and diffusion coefficients possess Lipschitz and growth conditions.

We then consider a non-homogeneous Poisson process with intensity function $\lambda:\mathbb{R}^{d_x}\rightarrow[\underline{c},\overline{c}]$, $0<\underline{c}<\overline{c}<+\infty$. That is to say that, conditional on the process $\{X_r\}_{r\in[0,t]}$, the joint density of $n_t\in\mathbb{N}$ event times $s_1,\dots,s_{n_t}$ is precisely
$$
\left(\prod_{k=1}^{n_t} \lambda(x_{s_k})\right)\exp\left\{-\int_{0}^t \lambda(x_r)dr\right\}.
$$
The observations, conditional upon the event times $s_{1:n_t}$ and the diffusion process  $\{X_r\}_{r\in[0,t]}$, have joint density
$$
p(y_{s_1:s_{n_t}}|s_{1:n_t},\{x_r\}_{r\in[0,t]}) = \prod_{k=1}^{n_t} g(x_{s_k},y_{s_k})
$$
where, for any $x\in\mathbb{R}^{d_x}$, $g(x,\cdot)$ is a probability density on $\mathsf{Y}$.

Our objective is to do filtering at some finite collection of times. For now, consider $t\in\mathbb{R}^+$ given, then, for $\varphi\in\mathcal{B}_b(\mathbb{R}^{d_x})$, we want to compute
\begin{equation}\label{eq:cont_filter}
\pi_t(\varphi) := \frac{
\mathbb{E}\left[\varphi(X_t)\left(\prod_{k=1}^{n_t} g(X_{s_k},y_{s_k})\lambda(X_{s_k})\right)
\exp\left\{-\int_{0}^t \lambda(X_r)dr\right\}
\right]
}{\mathbb{E}\left[\left(\prod_{k=1}^{n_t} g(X_{s_k},y_{s_k})\lambda(X_{s_k})\right)
\exp\left\{-\int_{0}^t \lambda(X_r)dr\right\}
\right]},
\end{equation}
where the expectation is w.r.t.~the law of the diffusion process given in \eqref{eq:diff}. We will make the constraint that this computation is of interest at times $t\in\mathbb{N}$.

\subsection{Discretization}

In practice, it is typically difficult (or impossible) to work directly with \eqref{eq:diff} and \eqref{eq:cont_filter}. Often one works with a time discretization of \eqref{eq:diff}, of which we adopt the Euler-Maruyama approximation, with $\Delta_l=2^{-l}$, $l\in\mathbb{N}_0$:
\begin{equation}\label{eq:euler}
\widetilde{X}_{(k+1)\Delta_l} = \widetilde{X}_{k\Delta_l} + b(\widetilde{X}_{k\Delta_l}) + \sigma(\widetilde{X}_{k\Delta_l})[W_{(k+1)\Delta_l}-W_{k\Delta_l}]
\end{equation}
where $k\in\mathbb{N}_0$ and $\widetilde{X}_0=x_{*}$. We make the standard extension that if $t\in[k\Delta_l,(k+1)\Delta_l)$ then
\begin{equation}\label{eq:inbetween_euler}
\widetilde{X}_{t}=\widetilde{X}_{k\Delta_l}+(\widetilde{X}_{(k+1)\Delta_l}-\widetilde{X}_{k\Delta_l})\Delta_l^{-1}(t-k\Delta_l).
\end{equation}
This is clearly needed given the representation \eqref{eq:cont_filter} that we seek to approximate.

Given the Euler approximation \eqref{eq:euler}, we can then consider the approximation, for $t\in\mathbb{N}$, given by
\begin{equation}\label{eq:cont_filter_euler}
\pi_t^l(\varphi) := \frac{
\mathbb{E}\left[\varphi(X_t)\left(\prod_{k=1}^{n_t} g(\widetilde{X}_{s_k},y_{s_k})\lambda(\widetilde{X}_{s_k})\right)
\exp\left\{-\Delta_l\sum_{k=0}^{t\Delta_l^{-1}-1}\lambda(\widetilde{X}_{k\Delta_l})\right\}
\right]
}{\mathbb{E}\left[\left(\prod_{k=1}^{n_t} g(\widetilde{X}_{s_k},y_{s_k})\lambda(\widetilde{X}_{s_k})\right)
\exp\left\{-\Delta_l\sum_{k=0}^{t\Delta_l^{-1}-1}\lambda(\widetilde{X}_{k\Delta_l})\right\}
\right]},
\end{equation}
where the expectation is taken w.r.t.~the law associated to the Euler approximation \eqref{eq:euler}.

We introduce an additional assumption
\begin{hypD}\label{hyp_diff:2}
We have
\begin{enumerate}
\item{$\lambda\in\textrm{Lip}_{\|\cdot\|_2}(\mathbb{R}^{d_x})$.}
\item{For any $y\in\mathsf{Y}$, $g(\cdot,y)\in\mathcal{B}_b(\mathbb{R}^{d_x})\cap\textrm{Lip}_{\|\cdot\|_2}(\mathbb{R}^{d_x})$.}
\end{enumerate}
\end{hypD}
We have the following result.
\begin{prop}\label{prop:weak_error}
Assume (D\ref{hyp_diff:1}-\ref{hyp_diff:2}). For $(t,\varphi)\in\mathbb{N}\times\mathcal{B}_b(\mathbb{R}^{d_x})\cap\textrm{\emph{Lip}}_{\|\cdot\|_2}(\mathbb{R}^{d_x})$ there exists a $C<+\infty$ such that for any $l\in\mathbb{N}_0$
$$
|\pi_t^l(\varphi)-\pi_t(\varphi)| \leq C\Delta_l.
$$
\end{prop}
\begin{proof}
One can see \cite{crisan_bain,kloeden} for justifications of this result as it is standard in the literature.
\end{proof}

\subsection{Filtering}

We shall now follow \cite{jasra} and present the filter in a recursive form. We will then detail some algorithms for approximating the afore-mentioned filter.

\subsubsection{The Filter Measure}\label{sec:disc_filter}

Throughout the section $l\in\mathbb{N}_0$ is given and as $x_{*}$ is fixed, it is removed from the notation where possible. This presentation closely follows \cite{jasra} and the purpose of these notations is to achieve a consistency with the literature on particle filters. Let $p\in\mathbb{N}_0$ be given, we use the notation
$$
u_p^l := (x_{p+\Delta_l},x_{p+2\Delta_l},\dots, x_{p+1}) \in\mathbb{R}^{d_x(\Delta_l^{-1})} =: E_l
$$
which denotes a path on a unit time $(p,p+1]$; we have removed the $\widetilde{\cdot}$ notation for simplicity. 
We now set for $p\in\mathbb{N}_0$
$$
G_p^l(u_{p-1}^l,u_p^l) := \left(\prod_{k=n_{p}+1}^{n_{p+1}}g(x_{s_k},y_{s_k})\lambda(x_{s_k})\right)\exp\left\{-\Delta_l\sum_{k=0}^{\Delta_l^{-1}-1}\lambda(x_{p+(k+1)\Delta_l})\right\},
$$
where $n_0=0$, $u_{-1}^l=x_{*}$ and this is a function on $E_{l-1} \times E_l$ (in fact it is a function on $\mathbb{R}^{d_x}\times E_l$ because only $x_p^l$ from $u_{p-1}^l$ is used as we will see later) due to the interpolation defined in \eqref{eq:inbetween_euler}. 
The density function of Euler-Maruyama discretizations can be characterized by Gaussian densities; we denote by $m^{l}$ the transition of \eqref{eq:euler}. For $p \in \mathbb{N}$, the initial measures and Markov kernels that we will need are
\begin{align*}
\eta_{0}^l(du_{0}^l) & =  \prod_{k=1}^{\Delta_l^{-1}} m^l(x_{(p-1)\Delta_l},x_{p\Delta_l})dx_{k\Delta_l} \\
\vspace{-5pt}
M^{l}(u_{p-1}^l,du_{p}^l) & = \prod_{k=1}^{\Delta_l^{-1}} m^l(x_{p+(k-1)\Delta_l},x_{p+k\Delta_l})dx_{p+k\Delta_l}.
\end{align*}
It is worth noting that the dependence of both $G_{p}^{l}(u_{p-1}^l,u_{p}^l)$ and $M^{l}(u_{p-1}^l,du_{p}^l)$ on $u_{p-1}^l$ is present only on its last element $x_p^l$, meaning that for $M^l$ the transition is only dependent on $x_p^l$, we will use these facts below.

The time-discretized filter \emph{on path space} (i.e.~over time interval $(p-1,p]$, $p\in\mathbb{N}$) can then be written, using the above notation, for any $p\in\{2,3,\dots\}$ as
\begin{equation}\label{eq:path_filt}
\overline{\pi}_{p}^l(du_{p-1}^l) := \frac{\int_{E_l^{p-1}}\left\{\prod_{k=0}^{p-1}G_k^l(u_{k-1}^l,u_k^l)\right\}\eta_{0}^l(du_{0}^l)\prod_{k=1}^{p-1} M^{l}(u_{k-1}^l,du_{k}^l)}{\int_{E_l^{p}}\left\{\prod_{k=0}^{p-1}G_k^l(u_{k-1}^l,u_k^l)\right\}\eta_{0}^l(du_{0}^l)\prod_{k=1}^{p-1} M^{l}(u_{k-1}^l,du_{k}^l)}
\end{equation}
with 
$$
\overline{\pi}_{1}^l(du_0^l) := \frac{G_0^l(u_{-1}^l,u_0^l)~\eta_{0}^l(du_{0}^l)}{\int_{E_l}G_0^l(u_{-1}^l,u_0^l)~\eta_{0}^l(du_{0}^l)}.
$$
This collection of probability measures will be of use later on in the article.
We note that, for any $p\in\mathbb{N}$, $l\in \mathbb{N}_0$, and $\varphi\in\mathcal{B}_b(\mathbb{R}^{d_x})$, we have that the expectation of $\varphi$ w.r.t. the approximate filter distribution at time $p$ as
$$
\pi_p^l(\varphi) = \int_{E_l}\varphi(x_p)\overline{\pi}_{p}^l(du_{p-1}^l).
$$

\subsubsection{Particle and Coupled Particle Filter}

We begin by describing the particle filter for approximating $\pi_p^l(\varphi)$ for a given $l\in\mathbb{N}_0$, with $p\in\mathbb{N}$ and any $\varphi:\mathbb{R}^{d_x}\rightarrow\mathbb{R}$ that is $\pi_p^l-$integrable. For a given $N\in\mathbb{N}$, the particle filter generates a system of random variables on $(E_l^N)^{n+1}$ at a time $n\in\mathbb{N}_0$ according to the probability measure
$$
\mathbb{Q}(d(u_0^{l,1:N},\dots,u_n^{l,1:N})) = \left\{\prod_{i=1}^N M^l(x_{*},du_0^{l,i})\right\}\prod_{p=0}^{n-1}\prod_{i=1}^N
\left\{\sum_{j=1}^{N}\frac{G_{p}^l(u_{p-1}^{l,j},u_p^{l,j})}{\sum_{k=1}^NG_{p}^l(u_{p-1}^{l,k},u_p^{l,k})} M^l(u_{p}^{l,j},du_{p+1}^{l,i})\right\}.
$$
An algorithmic description of the particle filter is given in \autoref{alg:pf}. For $(t,\varphi)\in\mathbb{N}\times\mathcal{B}_b(\mathbb{R}^{d_x})$ one can approximate the time-discretized filter $\pi_{t}^l(\varphi)$ corresponding to step-size $\Delta_l$ via
\begin{equation}\label{eq:pf_est}
\pi_{t}^{l,N}(\varphi) := \frac{\sum_{i=1}^N G_{t-1}^l(u_{t-2}^{l,i},u_{t-1}^{l,i})\varphi(x_t^{l,i})}{\sum_{i=1}^N G_{t-1}^l(u_{t-2}^{l,i},u_{t-1}^{l,i})},
\end{equation}
which can easily be shown to converge (e.g.~in probability as $N\rightarrow\infty$) to $\pi_{t}^l(\varphi)$; see \cite{delm:2004}.
We define an empirical measure at time $p-1$ and level $l$ which will be needed later on through the following expectation. Let $\varphi\in\mathcal{B}_b(E_l^2)$ and $p\in\mathbb{N}$, we define
\begin{equation}\label{eq:empirical_measure_pred}
\eta_{p-1}^{l,N}(\varphi) := \frac{1}{N}\sum_{i=1}^N \varphi(u_{p-2}^{l,i},u_{p-1}^{l,i}).
\end{equation}

\begin{algorithm}[h!]
\begin{enumerate}
\item{Initialize: For $i\in\{1,\dots,N\}$, generate $u_0^{l,i}$ from $M^l(x_{*},\cdot)$. Set $p=0$.}
\item{Update: For $i\in\{1,\dots,N\}$, generate $u_{p+1}^{l,i}$ from 
$$
\sum_{j=1}^{N}\frac{G_{p}^l(u_{p-1}^{l,j},u_p^{l,j})}{\sum_{k=1}^NG_{p}^l(u_{p-1}^{l,k},u_p^{l,k})} M^l(u_{p}^{l,j},du_{p+1}^{l,i}).
$$
Set $p=p+1$ and return to the start of 2.}
\end{enumerate}
\caption{Particle Filter.}
\label{alg:pf}
\end{algorithm}

We now consider the coupled particle filter (CPF) as developed in \cite{mlpf,mlpf_nc} (see also \cite{mlpf_new}). Contrary to the presentation in \cite{jasra} we restrict ourselves to an algorithmic rather than operator based description. The objective of the coupled particle filter, for $l\in\mathbb{N}$ given and up-to a time $t\in\mathbb{N}$, is to generate two clouds of particles $u_{t}^{l,1:N}$ and $\bar{u}_{t}^{l-1,1:N}$ so that $u_{t}^{l,1:N}$ (resp.~$\bar{u}_{t}^{l-1,1:N}$) can be used to approximate $\pi_t^l(\varphi)$ (resp.~$\pi_t^{l-1}(\varphi)$). Moreover, that there is a dependence between these two clouds of particles which ensure that the variance of terms which approximate the difference $\pi_t^l(\varphi)-\pi_t^{l-1}(\varphi)$ will
fall with $l$. 

To describe the CPF in its most simple form, we need two algorithms: coupled sampling and coupled resampling.
We begin with the former, which is given in \autoref{alg:coup_euler} and simply is the well-known synchronous coupling of Euler-discretized diffusions. The method of coupled resampling, developed for multilevel applications in \cite{mlpf} can be found in \autoref{alg:max_coup}. This resampling algorithm maximizes the probability that two sampled indices are equal. As noted in \cite{jasra_clt}, \autoref{alg:max_coup} is by no means optimal, but appears to be the most used method in the literature.

We now give the CPF algorithm in \autoref{alg:cpf}. 
This is simply a type of particle filter that utilizes \autoref{alg:coup_euler} for sampling and \autoref{alg:max_coup} for resampling.
The algorithm can estimate the difference
$\pi_t^l(\varphi)-\pi_t^{l-1}(\varphi)$, $t\in \mathbb{N}$, using the expression
\begin{equation}\label{eq:cpf_est}
[\pi_t^l-\pi_t^{l-1}]^N(\varphi) := 
\frac{\sum_{i=1}^N G_{t-1}^l(u_{t-2}^{l,i},u_{t-1}^{l,i})\varphi(x_t^{l,i})}{\sum_{i=1}^N G_{t-1}^{l}(u_{t-2}^{l,i},u_{t-1}^{l,i})} -
\frac{\sum_{i=1}^N G_{t-1}^{l-1}(\bar{u}_{t-2}^{l-1,i},\bar{u}_{t-1}^{l-1,i})\varphi(\bar{x}_t^{l-1,i})}{\sum_{i=1}^N G_{t-1}^{l-1}(\bar{u}_{t-2}^{l-1,i},\bar{u}_{t-1}^{l-1,i})}.
\end{equation}
We define two empirical measures at time $p-1$ and levels $l$, $l-1$ through the following expectations. For $(\varphi_1,\varphi_2)\in\mathcal{B}_b(E_l^2)\times\mathcal{B}_b(E_{l-1}^2)$, $p\in\mathbb{N}$, we have
\begin{equation}\label{eq:pred_cpf}
\check{\eta}_{p-1}^{l,N}(\varphi_1) := \frac{1}{N}\sum_{i=1}^N\varphi_1(u_{t-2}^{l,i},u_{t-1}^{l,i})\quad\textrm{and}\quad
\check{\bar{\eta}}_{p-1}^{l-1,N}(\varphi_2) := \frac{1}{N}\sum_{i=1}^N\varphi_2(\bar{u}_{t-2}^{l-1,i},\bar{u}_{t-1}^{l-1,i}).
\end{equation}

\begin{algorithm}[h!]
\begin{enumerate}
\item{Input: level $l$ and starting points $(x_0^l,x_0^{l-1})$.}
\item{Generate $Z_k\stackrel{\textrm{i.i.d.}}{\sim}\mathcal{N}_d(0,\Delta_l I_d)$, $k\in\{1,2,\dots,\Delta_l^{-1}\}$.}
\item{Level $l$: For $k\in\{0,1,2,\dots,\Delta_l^{-1}-1\}$ with $X_0^l=x_0^l$ generate
$$
X_{(k+1)\Delta_l}^l = X_{k\Delta_l}^l + b(X_{k\Delta_l}^l)\Delta_l + \sigma(X_{k\Delta_l}^l)Z_{k+1}.
$$
Set $U_0^l = (X_{\Delta_l}^l,X_{2\Delta_l}^l,\dots, X_1^l)$.
}
\item{Level $l-1$: For $k\in\{0,1,2,\dots,\Delta_{l-1}^{-1}-1\}$ with $X_0^{l-1}=x_0^{l-1}$ generate
\begin{eqnarray*}
X_{(k+1)\Delta_{l-1}}^{l-1} & = & X_{k\Delta_{l-1}}^{l-1} + b(X_{k\Delta_{l-1}}^{l-1})\Delta_{l-1} + \sigma(X_{k\Delta_{l-1}}^{l-1})\{Z_{2(k+1)-1}+Z_{2(k+1)}\}.
\end{eqnarray*}
}
Set $\overline{U}_0^{l-1} = (X_{\Delta_{l-1}}^{l-1},X_{2\Delta_{l-1}}^{l-1},\dots, X_1^{l-1})$.
\item{Output: $(U_0^l,\overline{U}_0^{l-1})$.}
\end{enumerate}
\caption{Coupled Euler Scheme on $[0,1]$.}
\label{alg:coup_euler}
\end{algorithm}

\begin{algorithm}[h!]
\begin{enumerate}
\item{Input: $N\in\mathbb{N}$, two clouds of particles $(U_1^{1:N},U_2^{1:N})$ and their associated probabilities $(W_1^{1:N},W_2^{1:N})$.}
\item{For $i\in\{1,\dots,N\}$, generate $R^i\sim\mathcal{U}_{[0,1]}$
\begin{itemize}
\item{If $R^i<\sum_{i=1}^N \min\{W_1^i,W_2^i\}$, generate $a^i\in\{1,\dots,N\}$ using the probability mass function
$$
\mathbb{P}(i) = \frac{\min\{W_1^i,W_2^i\}}{\sum_{j=1}^N \min\{W_1^j,W_2^j\}}
$$
and set $\tilde{U}_j^i=U_j^{a^i}$, 
$j\in\{1,2\}$.}
\item{Otherwise generate $(a_1^i,a_2^i)\in\{1,\dots,N\}^2$ using any coupling of the probability mass functions:
$$
\mathbb{P}_j(i) = \frac{W_j^i-\min\{W_1^i,W_2^i\}}{\sum_{k=1}^N[W_j^k-\min\{W_1^k,W_2^k\}]}
$$
and set $\tilde{U}_j^i=U_j^{a_j^i}$,  $j\in\{1,2\}$.}
\end{itemize}
}
\item{Set: $U_j^i=\tilde{U}_j^{i}$, $(i,j)\in\{1,\dots,N\}\times\{1,2\}$.}
\item{Output: $(U_1^{1:N},U_2^{1:N})$.}
\end{enumerate}
\caption{Maximal Coupling Resampling}
\label{alg:max_coup}
\end{algorithm}

\begin{algorithm}[h!]
\begin{enumerate}
\item{Input: $(l,N)\in\mathbb{N}^2$.}
\item{Initialize: For $i\in\{1,\dots,N\}$, generate $u_0^{l,i},\bar{u}_0^{l-1,i}$ by using \autoref{alg:coup_euler} with level $l$ and starting points $(x_{*},x_{*})$. Set $p=1$.}
\item{Iterate: For $i\in\{1,\dots,N\}$, compute the weights
$$
w_{p-1}^{l,i} = \frac{G_{p-1}^l(u_{p-2}^{l,i},u_{p-1}^{l,i})}{\sum_{j=1}^NG_{p-1}^l(u_{p-2}^{l,j},u_{p-1}^{l,j})} \quad\textrm{and}\quad \bar{w}_{p-1}^{l-1,i} = \frac{G_{p-1}^{l-1}(
\bar{u}_{p-2}^{l-1,i},\bar{u}_{p-1}^{l-1,i})}{\sum_{j=1}^NG_{p-1}^{l-1}(\bar{u}_{p-2}^{l-1,j},\bar{u}_{p-1}^{l-1,j})}.
$$
Perform \autoref{alg:max_coup} with inputs $N$, $(u_{p-1}^{l,1:N},\bar{u}_{p-1}^{l-1,1:N})$ and $(w_{p-1}^{l,1:N}, \bar{w}_{p-1}^{l-1,1:N})$. For $i\in\{1,\dots,N\}$, generate $u_p^{l,i},\bar{u}_p^{l-1,i}$ by using \autoref{alg:coup_euler} with level $l$ and starting points $(x_{p}^{l,i},\bar{x}_{p}^{l-1,i})$. Set $p=p+1$ and return to the start of 3.}
\end{enumerate}
\caption{Coupled Particle Filter}
\label{alg:cpf}
\end{algorithm}

\subsubsection{Multilevel Particle Filter}
We can now describe the MLPF using the PF and the CPF.
\begin{enumerate}
\item{Level 0: Run a PF as in \autoref{alg:pf} with $N_0$ samples, independently of all other levels.} 
\item{For each level $l\in\{1,\dots,L\}$: Run a CPF (to approximate the time-discretized filters at levels $l$ and $l-1$) as in \autoref{alg:cpf} with $N_l$ samples, independently of all other levels.}
\end{enumerate}
An estimator of expectations $\pi_{t}^{L}(\varphi)$ for $t\in\mathbb{N}, \varphi\in\mathcal{B}_b(\mathbb{R}^{d_x})$ with respect to the time-discretized filter at the highest level $L$ is then given by
\begin{equation}\label{eq:ml_est}
\pi_{t}^{L,ML}(\varphi) := \pi_{t}^{0,N_0}(\varphi) + \sum_{l=1}^L [\pi_{t}^{l}-\pi_{t}^{l-1}]^{N_l}(\varphi),
\end{equation}
where $\pi_{t}^{0,N_0}(\varphi)$ is the estimator in \eqref{eq:pf_est} at level $l=0$ with $N=N_0$ samples, 
and $[\pi_{t}^{l}-\pi_{t}^{l-1}]^{N_l}(\varphi)$ is the estimator in \eqref{eq:cpf_est} with $N=N_l$ samples.

\subsubsection{Unbiased Particle Filter}

The method of \cite{ub_pf} can be used to remove the time-discretization bias and we detail the method with as much brevity as possible; details can be found in the afore-mentioned reference.

The basic method in \cite{ub_pf} uses a double randomization technique which is based upon the approaches in \cite{mcl,rhee} (see also \cite{vihola}). 
Let $\varphi\in\mathcal{B}_b(\mathbb{R}^{d_x})$ be fixed but arbitrary and note that the idea below can be used with little effort for many different $\varphi$.
This constitutes selecting a positive probability mass function $\mathbb{P}_L$ on $\mathbb{N}_0$ and
constructing a sequence of independent random variables $(\Xi_l)_{l\in\mathbb{N}_0}$ so that for any $(l,t)\in\mathbb{N}^2$:
$$
\mathbb{E}[\Xi_l] = \pi_t^l(\varphi) - \pi_t^{l-1}(\varphi)
$$
with $\mathbb{E}[\Xi_0] = \pi_t^0(\varphi)$ and $\mathbb{E}$ is the expectation w.r.t.~the law associated to the simulation of the algorithm. Then one can sample $L$ from $\mathbb{P}_L$ and then construct $\Xi_L$ to obtain
the estimator 
\begin{equation}\label{eq:ub_est}
\widehat{\pi_t(\varphi)} = \frac{\Xi_L}{\mathbb{P}_L(L)},
\end{equation}
which is an unbiased estimator of $ \pi_t(\varphi)$ and moreover, if one selects $\mathbb{P}_L$ so that
$$
\sum_{l\in\mathbb{N}_0}\frac{\mathbb{E}[\Xi_l^2]}{\mathbb{P}_L(l)} <+\infty,
$$
then the variance of the estimator in \eqref{eq:ub_est} is finite.
\cite{ub_pf} provides a method for constructing the sequence $(\Xi_l)_{l\in\mathbb{N}_0}$  and we now describe this.

We begin with the computation of $\Xi_0$ and the procedure that is needed is detailed in \autoref{alg:xi_0}.
The method requires a positive conditional probability mass function $\mathbb{P}_p(\cdot|l)$ \footnote{An independent probability mass function is also possible; see Section \ref{sec:ub_pf_numer}.} on $\mathbb{N}_0$ and a sequence
$(N_p)_{p\geq 0}$ of non-decreasing positive integers such that $N_p\uparrow\infty$ (e.g.~$N_p=N_0~ 2^p$). To compute
$\Xi_0$ we set, for $(p,t,\varphi)\in\mathbb{N}_0\times\mathbb{N} \times\mathcal{B}_b(E_0^2)$ given,
$$
\widetilde{\eta}_{t-1}^{0,N_{0:p}}(\varphi) := \sum_{q=0}^p \left(\frac{N_q-N_{q-1}}{N_p}\right)\eta_{t-1}^{0,N_q-N_{q-1}}(\varphi)
$$
with $N_{-1}=0$.
Then set for $\varphi\in\mathcal{B}_b(\mathbb{R}^{d_x})$
$$
\widetilde{\pi}_t^{0,N_{0:p}}(\varphi) := \frac{\int_{E_{0}^2} \varphi(x_t^0)~G_{t-1}^0(u_{t-2}^0,u_{t-1}^0) ~\widetilde{\eta}_{t-1}^{0,N_{0:p}}\left(d(u_{t-2}^0,u_{t-1}^0)\right)}
{\int_{E_{0}^2} G_{t-1}^0(u_{t-2}^0,u_{t-1}^0)~ \widetilde{\eta}_{t-1}^{0,N_{0:p}}\left(d(u_{t-2}^0,u_{t-1}^0)\right)}.
$$
Then we have
\begin{equation}\label{eq:xi_0}
\Xi_0 = \frac{1}{\mathbb{P}(p|l=0)}\left(\widetilde{\pi}_t^{0,N_{0:p}}(\varphi)-\widetilde{\pi}_t^{0,N_{0:p-1}}(\varphi)\right)
\end{equation}
with the convention that $\widetilde{\pi}_t^{0,N_{0:-1}}(\varphi)=0$.

\begin{algorithm}[h!]
\begin{enumerate}
\item{Sample $P$ from $\mathbb{P}_p(\cdot|l)$.}
\item{Run \autoref{alg:pf} with $N_0$ samples until time $t-1$ and denote the empirical measure (as in \eqref{eq:empirical_measure_pred}) as $\eta_{t-1}^{l,N_0}(\cdot)$. Set $q=1$. If $p=0$ stop; otherwise go to the next step.}
\item{Independently of all other random variables, run \autoref{alg:pf} with $N_q-N_{q-1}$ samples until time $t-1$ and denote the empirical measure (as in \eqref{eq:empirical_measure_pred}) as $\eta_{t-1}^{l,N_q-N_{q-1}}(\cdot)$. Set
$q=q+1$. If $q=p+1$ stop; otherwise go to the start of 3.}
\end{enumerate}
\caption{Computing $\Xi_0$. Throughout $(t,l)\in\mathbb{N}\times\mathbb{N}_0$ are given.}
\label{alg:xi_0}
\end{algorithm}

For the computation of $(\Xi_l)_{l\in\mathbb{N}}$, the procedure is similar, except one uses a CPF instead of a PF; see \autoref{alg:xi_l}. We set, for $(p,t,\varphi_1)\in\mathbb{N}_0\times\mathbb{N}\times\mathcal{B}_b(E_l^2)$ given,
$$
\widetilde{\eta}_{t-1}^{l,N_{0:p}}(\varphi_1) := \sum_{q=0}^p \left(\frac{N_q-N_{q-1}}{N_p}\right)\check{\eta}_{t-1}^{l,N_q-N_{q-1}}(\varphi_1)
$$
and, for $(p,t,\varphi_2)\in\mathbb{N}_0\times\mathbb{N}\times\mathcal{B}_b(E_{l-1}^2)$ given,
$$
\widetilde{\bar{\eta}}_{t-1}^{l-1,N_{0:p}}(\varphi_2) := \sum_{q=0}^p \left(\frac{N_q-N_{q-1}}{N_p}\right)\check{\bar{\eta}}_{t-1}^{l-1,N_q-N_{q-1}}(\varphi_2).
$$
Then set for $\varphi\in\mathcal{B}_b(\mathbb{R}^{d_x})$
\begin{align*}
[\widetilde{\pi}_t^{l}&-\widetilde{\pi}_t^{l-1}]^{N_{0:p}}(\varphi)=\\
&\frac{ \int_{E_{l}^2} \varphi(x_t^l)G_{t-1}^l(u_{t-2}^l,u_{t-1}^l)\widetilde{\eta}_{t-1}^{l,N_{0:p}}\left(d(u_{t-2}^l,u_{t-1}^l)\right) }
{\int_{E_{l}^2} \varphi(x_t^l)G_{t-1}^l(u_{t-2}^l,u_{t-1}^l)\widetilde{\eta}_{t-1}^{l,N_{0:p}}\left(d(u_{t-2}^l,u_{t-1}^l)\right)} -
\frac{\int_{E_{l-1}^2} \varphi(x_t^{l-1})G_{t-1}^{l-1}(\bar{u}_{t-2}^{l-1},\bar{u}_{t-1}^{l-1})\widetilde{\bar{\eta}}_{t-1}^{l-1,N_{0:p}}\left(d(\bar{u}_{t-2}^{l-1},\bar{u}_{t-1}^{l-1})\right)}
{\int_{E_{l-1}^2} G_{t-1}^{l-1}(\bar{u}_{t-2}^{l-1},\bar{u}_{t-1}^{l-1})\widetilde{\bar{\eta}}_{t-1}^{l-1,N_{0:p}}\left(d(\bar{u}_{t-2}^{l-1},\bar{u}_{t-1}^{l-1})\right)}.
\end{align*}
Then we have
\begin{equation}\label{eq:xi_l}
\Xi_l = \frac{1}{\mathbb{P}(p|l)}\Big([\widetilde{\pi}_t^{l}-\widetilde{\pi}_t^{l-1}]^{N_{0:p}}(\varphi)-[\widetilde{\pi}_t^{l}-\widetilde{\pi}_t^{l-1}]^{N_{0:p-1}}(\varphi)\Big)
\end{equation}
with the convention that $[\widetilde{\pi}_t^{l}-\widetilde{\pi}_t^{l-1}]^{N_{0:-1}}(\varphi)=0$.

\begin{algorithm}[H]
\begin{enumerate}
\item{Sample $P$ from $\mathbb{P}_p(\cdot|l)$.}
\item{Run \autoref{alg:cpf} with $N_0$ samples until time $t-1$ and 
denote the empirical measures 
(as in \eqref{eq:pred_cpf}) as $\check{\eta}_t^{l,N_0}(\cdot)$ and $\check{\bar{\eta}}_{t-1}^{l-1,N_0}(\cdot)$. Set $q=1$. If $p=0$ stop; otherwise go to the next step.}
\item{Independently of all other random variables, run \autoref{alg:cpf} with $N_q-N_{q-1}$ samples and denote the 
empirical measures (as in \eqref{eq:pred_cpf}) as $\check{\eta}_{t-1}^{l,N_q-N_{q-1}}(\cdot)$ and $\check{\bar{\eta}}_{t-1}^{l-1,N_q-N_{q-1}}(\cdot)$. Set $q=q+1$. If $q=p+1$ stop; otherwise go to the start of 3.}
\end{enumerate}
\caption{Computing $\Xi_l$. Throughout $(t,l)\in\mathbb{N}\times\mathbb{N}$ are given.}
\label{alg:xi_l}
\end{algorithm}

Then to compute the unbiased estimator, one can simply use the expression in \eqref{eq:ub_est}, repeated $M\in\mathbb{N}$ times and in parallel; i.e.~a Monte Carlo estimator
\begin{align}
\label{eq:ub_est_M}
\widehat{\widehat{\pi_t(\varphi)}} := \frac{1}{M}\sum_{i=1}^M \frac{\Xi_{L_i}}{\mathbb{P}_L(L_i)},
\end{align}
 see \cite[Algorithm 5]{ub_pf} for further details.
Note that the finite variance and unbiasedness is to be established in Section \ref{sec:math}.

\subsection{Parameter Estimation}\label{sec:online}
\subsubsection{Model with Static Parameters and Score Function}

We consider the model as in Section \ref{sec:mod} except with some static parameters $\theta\in\Theta$.
For instance, the diffusion process $\{X_t\}_{t\geq 0}$ depends on a static parameter $\theta\in\Theta$ in the following manner
\begin{equation}
dX_t=b_{\theta}(X_t)dt+\sigma(X_t)dW_t
\label{eq:diff_theta}
\end{equation}
such that $X_0=x_*\in \mathbb{R}^{d_x}$ is given, $b_{\theta}:\mathbb{R}^{d_x}\rightarrow \mathbb{R}^{d_x}$ $   
\forall \theta \in \Theta$, $\sigma:\mathbb{R}^{d_x}\rightarrow \mathbb{R}^{d_x\times d_x}$ is non-constant, and $\{W_t\}_{t\geq 0}$ is a standard Brownian motion. In the remainder of the model (i.e.~as in Section \ref{sec:mod}) we allow both $g$ and $\lambda$ to depend on $\theta$ as well, using the subscript $\theta$ in the notation from herein. 

We define, for $\varphi\in \mathcal{B}_b(\mathbb{R}^{d_x})$, $t\in\mathbb{N}$
\begin{align*}
\gamma_{t,\theta}(\varphi):=\mathbb{E}_{\theta}\left[\varphi(X_t)\left(\prod_{k=1}^{n_t} g_\theta(X_{s_k},y_{s_k})\lambda_\theta(X_{s_k})\right)
\exp\left\{-\int_{0}^t \lambda_\theta(X_r)dr\right\} \right]
\end{align*}
where the expectation $\mathbb{E}_{\theta}$ is w.r.t.~law of the process \eqref{eq:diff_theta}. Notice that for $\varphi(x)\equiv 1$, this represents the likelihood of the data. 
Our objective is to perform an online (as the data arrive) static parameter estimation using the score function which (under simple assumptions which are not stated here), for $t\in\mathbb{N}$, is given by
\begin{align}\label{eq:score}
\nabla_\theta \log \left(\gamma_{t, \theta}(1)\right) = \frac{
\mathbb{E}_{\theta}\left[\zeta_{t,\theta}\left(\prod_{k=1}^{n_t} g_{\theta}(X_{s_k},y_{s_k})\lambda_{\theta}(X_{s_k})\right)
\exp\left\{-\int_{0}^t \lambda_{\theta}(X_r)dr\right\}
\right]
}{\mathbb{E}_{\theta}\left[\left(\prod_{k=1}^{n_t} g_{\theta}(X_{s_k},y_{s_k})\lambda_{\theta}(X_{s_k})\right)
\exp\left\{-\int_{0}^t \lambda_{\theta}(X_r)dr\right\}
\right]},
\end{align}
where 
\begin{align}
\zeta_{t,\theta}:=\int_0^t\left(\nabla_\theta b_\theta\left(X_r\right)\right)^* a\left(X_r\right)^{-1} \sigma\left(X_r\right) d W_r+\sum_{k=1}^{n_t}\nabla_\theta \log\left( g_{\theta}(X_{s_k},y_{s_k})\lambda_{\theta}(X_{s_k})\right)-\int_{0}^t \nabla_\theta \lambda_\theta(X_r)dr.
\label{eq:omega}
\end{align}
An analogous derivation of \eqref{eq:score} can be found in Appendix A of \cite{beskos}.

\subsubsection{Stochastic Gradient Approach}
In our parameter estimation strategy, we attempt to maximize the limiting \textit{average log-likelihood}
\begin{align*}
\mathcal{L}(\theta)=\lim_{t\rightarrow \infty}\frac{1}{t}\int_{0}^t \log \gamma_{s,\theta}(1)ds
\end{align*}
which, as well as its gradient $\nabla \mathcal{L}(\theta)$, can be shown to be an ergodic average under appropriate stability and regularity conditions \cite{surace}. One can estimate the parameters using the gradient $\nabla_\theta \log \left(\gamma_{t, \theta}(1)\right)$ and a stochastic gradient ascent methodology as
\begin{align}
\theta_{m+1}=\theta_{m}+\alpha_m  \nabla_\theta \log \left(\gamma_{m+1, \theta_{0:m}}(1)\right),
\label{eq:sgd_1}
\end{align}
where $m\in \mathbb{N}_0$, $\theta_0 \in \mathbb{R}^{d_\theta}$ is given, and $\{\alpha_m\}_{m \in \mathbb{N}}$ is a decreasing sequence of step-sizes. The recursion \eqref{eq:sgd_1} is not suitable for online computations since the complexity of the score function is $\mathcal{O}(m)$. In practice, we use online updates of the score function $\nabla_\theta \log \left(\gamma_{m+1, \theta_{0:m}}(1)\right)$ that include all updates of the parameters up to time $m-1$, where the filter in the interval of time $(m,m+1]$ is updated using the parameter $\theta_m$, see \cite{legland, kantas} (more details in Section \ref{subsec:par_est}). This approach can be counterproductive in the sense that the first parameters $\theta_0$, $\theta_1,\cdots, \theta_{q}$, $q<<m$ 
are always included in the score function. In order to alleviate this problem we use a different estimator $\nabla_\theta \log \left(\gamma_{c(m+1), \theta_{0:m}}(1)\right)-\nabla_\theta \log \left(\gamma_{cm, \theta_{0:m-1}}(1)\right)$, $c\in \mathbb{N}$, that highlights the new observations. Its stochastic gradient ascent equation is
\begin{align}
\theta_{m+1}=\theta_{m}+\alpha_m  \left[\nabla_\theta \log \left(\gamma_{c(m+1), \theta_{0:m}}(1)\right)-\nabla_\theta \log \left(\gamma_{cm, \theta_{0:m-1}}(1)\right)\right]
\label{eq:sgd}
\end{align}
with the convention $\nabla_\theta \log \left(\gamma_{0, \theta_{0:-1}}(1)\right)=0$. It can be shown, for $\{\alpha_m\}_{m\in \mathbb{N}_0}$ such that $\alpha_m\to 0$, $\sum_{m \in \mathbb{N}_0}\alpha_m=\infty$ and $\sum_{m \in \mathbb{N}_0}\alpha_m^2 <\infty$, that $\theta_m$ in \eqref{eq:sgd_1} and \eqref{eq:sgd} converges as $m\rightarrow \infty$, for proofs see \cite{BMP90,legland}. In the following sections, we study the particle systems of discretized estimation of the score function.

\subsubsection{Discretized Score}

In order to compute practical approximations of the score, we resort to the same discretization scheme displayed in \eqref{eq:euler} and \eqref{eq:cont_filter_euler} where we replace the functions $b,$ $g$ and $\lambda$ with its $\theta$ dependent counterparts. The time discretization of $\zeta_{t,\theta}$ is written as
\begin{align*}
\zeta_{t, \theta}^l\left(x_0, x_{\Delta_l}, \ldots, x_t\right):=&\sum_{k=0}^{t\Delta_l^{-1}-1}\left\{\left(\nabla_\theta b_\theta\left(x_{k \Delta_l}\right)\right)^* a\left(x_{k \Delta_l}\right)^{-1} \sigma\left(x_{k \Delta_l}\right)\left(W_{(k+1) \Delta_l}-W_{k \Delta_l}\right)\right.\\
&+\sum_{k=1}^{n_t}\nabla_\theta \log\left( g_{\theta}(x_{s_k},y_{s_k})\lambda_{\theta}(x_{s_k})\right)-\sum_{k=0}^{t\Delta_l^{-1}-1} \nabla_\theta \lambda_\theta(x_{k\Delta_l})\Delta_l
\end{align*}
and recall the interpolation \eqref{eq:inbetween_euler}. Now, we approximate the score function as 
\begin{align*}
\nabla_\theta \log \left(\gamma_{t, \theta}^l(1)\right):=\frac{{\mathbb{E}}_\theta\left[\zeta_{t, \theta}^l\left(\widetilde{X}_0, \widetilde{X}_{\Delta_l}, \ldots, \widetilde{X}_t\right)\left(\prod_{i=1}^{n_t} g_{\theta}(\widetilde{X}_{s_i},y_{s_i})\lambda_{\theta}(\widetilde{X}_{s_i})\right)\exp\left\{-\sum_{j=0}^{t\Delta_l^{-1}-1} \lambda_{\theta}(\widetilde{X}_{j\Delta_l})\Delta_l\right\}\right] }{{\mathbb{E}}_\theta\left[  \left(\prod_{i=1}^{n_t} g_{\theta}(\widetilde{X}_{s_i},y_{s_i})\lambda_{\theta}(\widetilde{X}_{s_i})\right)\exp\left\{-\sum_{j=0}^{t\Delta_l^{-1}-1} \lambda_{\theta}(\widetilde{X}_{j\Delta_l})\Delta_l\right\}   \right]}.
\end{align*}
Following the proof of \cite[Appendix B]{beskos}, one can establish convergence of $\nabla_\theta \log \left(\gamma_{t, \theta}^l(1)\right)$ to $\nabla_\theta \log \left(\gamma_{t, \theta}(1)\right)$ as $l\rightarrow\infty$; we do not give the statement here as it is essentially analogous to that in \cite{beskos}.

\subsubsection{Backward Feynman-Kac Model and Particle Smoothing}

In the following, we define some objects to help us build the particle smoother which will ultimately allow us to define an online algorithm for estimating the score function.
We shall suppose that $t\in\mathbb{N}$ is fixed in the forthcoming description.


Now we can represent the discretized smoother for $p\in\mathbb{N}$ (see \eqref{eq:path_filt} also)
\begin{align}
\label{eq:FKmodel}
\overline{\pi}_{p,\theta}^{l,S}\big(d(u_{0}^l,\dots,u_{p-1}^l)\big) :=  \frac{\big(\prod_{k=0}^{p-1} G_{k,\theta}^l(u_{k-1}^l,u_{k}^l)\big)\,\eta_{0,\theta}^l(du_{0}^l)\prod_{k=1}^{p-1} M_{\theta}^{l}(u_{k-1}^l,du_{k}^l)}
{\int_{E_l^{p}}\big(\prod_{k=0}^{p-1} G_{k,\theta}^l(u_{k-1}^l,u_{k}^l)\big)\,\eta_{0,\theta}^l(du_{0}^l)\prod_{k=1}^{p-1} M_{\theta}^{l}(u_{k-1}^l,du_{k}^l)}.
\end{align}
where we have added $\theta$ subscripts for the quantities $\eta_0^l$, $M^l$ and $G_k^l$ as in Section \ref{sec:disc_filter}.
As in the previous subsection, the Feynman-Kac structure of \eqref{eq:FKmodel} allows a particle filter estimation. We aim to estimate \eqref{eq:score} online using backward smoothing; this is possible given the additive structure of \eqref{eq:omega} in terms of the dependence of the variables $u_{k}^l$ \cite{del_fb}. Such dependence will be shown in the following. For any $p \in \mathbb{N}_0$, let $f_{p,\theta}^l: \mathbb{R}^{d_x\times d_x}\rightarrow \mathbb{R}^{d_\theta}$ be defined as
\begin{align*}
f_{p,\theta}^l(x_{p\Delta_l},x_{(p+1)\Delta_l})  := &  (\nabla_{\theta}b_{\theta}(x_{p\Delta_l}))^*a(x_{p\Delta_l})^{-1}\sigma(x_{p\Delta_l})(W_{(p+1)\Delta_l}-W_{p\Delta_l})  \\[0.2cm]  &+\sum_{k=n_{p\Delta_l}+1}^{n_{(p+1)\Delta_l}} \nabla_\theta \log\left( g_{\theta}({x}_{s_k},y_{s_k})\lambda_{\theta}({x}_{s_k})\right)- \nabla_\theta \lambda_\theta(x_{p\Delta_l})\Delta_l.
\end{align*}
Now we define for $p\in \{0,\dots,t-1\}$ so that
\begin{align}
\mu_{p,\theta}^l(u_{p-1}^l,u_{p}^l) &:= \sum_{k=0}^{\Delta_l^{-1}-1}f_{k,\theta}^l(x_{p+k\Delta_l},x_{p+(k+1)\Delta_l});
\nonumber \\
F_{t,\theta}^l(u_{0}^l,\dots,u_{t-1}^l) &:= \sum_{p=0}^{t-1}\mu_{p,\theta}^l(u_{p-1}^l,u_{p}^l)\,\,\Big( \equiv \zeta_{t,\theta}^l(x_0,x_{\Delta_l},\dots,x_t)  \Big). \label{eq:defF}
\end{align}

The discretized smoothing distribution can be written via the time-reversal formula for hidden Markov models (see e.g.~\cite[Section 3.2]{beskos}) as
\begin{align*}
\overline{\pi}_{t,\theta}^{l,S}\big(d(u_{0}^l,\dots,u_{t-1}^l) \big) := \overline{\pi}_{t,\theta}^l(du_{t-1}^l)
\prod_{k=1}^{t-1}  B_{k,\theta,\overline{\pi}_{k,\theta}^l}^l(u_{k}^l,du_{k-1}^l),
\end{align*}
where $\overline\pi_{t,\theta}^l$ is the path-wise filter as in \eqref{eq:path_filt} with $\theta$ subscripts
and the backward Markov kernel is defined as, where we write the density of $M_{\theta}^l$ as $M_{\theta,d}^l$,
\begin{align}
\label{eq:BACK}
B_{k,\theta,\overline{\pi}_{k,\theta}^l}^l(u_{k}^l, &du_{k-1}^l) :=  
\frac{\overline{\pi}_{k,\theta}^l(du_{k-1}^l)\, 
G_{k,\theta}^l(u_{k-1}^l,u_{k}^l) M_{\theta,d}^l(u_{k-1}^l,u_{k}^l)} {\overline{\pi}_{k,\theta}^l(
G_{k,\theta}^l(\cdot,u_{k}^l) M_{\theta,d}^l(\cdot,u_{k}^l))}
\end{align}
with the notation $$\overline{\pi}_{k,\theta}^l(
G_{k,\theta}^l(\cdot,u_{k}^l) M_{\theta,d}^l(\cdot,u_{k}^l)) := \int_{E_l}\overline{\pi}_{k,\theta}^l(du_{k-1}^l) 
G_{k,\theta}^l(u_{k-1}^l,u_{k}^l) M_{\theta,d}^l(u_{k-1}^l,u_{k}^l).
$$
Given the structure of both $G_{k,\theta}^l$ and $M_{\theta,d}^l$, instead of depending on $u_{k}^l$, the backward kernel depends only on its first element $x_{k+\Delta_l}$, and the cost of computing it does not depend on the level $l$. For $k \in\left\{0,1, \ldots, t\Delta_l^{-1}-1\right\}$ we define
$$
\overline{g}_{k, \theta}^l\left(x_{k \Delta_l},x_{(k+1) \Delta_l}\right):=  \left(\prod_{i=n_{k\Delta_l}+1}^{n_{(k+1)\Delta_l}} g_{\theta}({x}_{s_i},y_{s_i})\lambda_{\theta}({x}_{s_i})\right)
\exp\left\{- \lambda_{\theta}(x_{k\Delta_l})\Delta_l\right\},
$$
then, one has
\begin{align}
B_{k,\theta,\overline{\pi}_{k,\theta}^l}^l(u_{k}^l, &du_{k-1}^l) :=  
\frac{\overline{\pi}_{k,\theta}^l(du_{k}^l)\, 
\overline{g}_{k,\theta}^l(x_{k},x_{k+\Delta_l}) m_{\theta}^l(x_{k},x_{k+\Delta_l})} {\overline{\pi}_{k,\theta}^l(
\overline{g}_{k,\theta}^l(\cdot,x_{k+\Delta_l}) m_{\theta}^l(\cdot,x_{k+\Delta_l}))},
\end{align}  
where $m_\theta^l$ is the transition of the Euler-Maruyama discretization of equation \eqref{eq:diff_theta}. Now, we have the following representation of the (time-discretized) score function
\begin{align}
\label{eq:backward_grad}
\nabla_{\theta}\log(\gamma_{t,\theta}^l(1)) = \int_{E_l^{t}}
F_{t,\theta}^l(u_{0}^l,\dots,u_{t-1}^l)\, \overline{\pi}_{t-1,\theta}^{l,S}\big(d(u_{0}^l,\dots,u_{t}^l) \big).
\end{align}
We can exploit the structure of \eqref{eq:defF} to estimate \eqref{eq:backward_grad}, which is partially our objective. The estimation is made online by computing particle estimators of the backward kernel. For precise details of the online smoothing methods see \cite{del_fb} for instance. 
The sequential estimation of the score function is detailed in \autoref{alg:online_disc}.  

\begin{algorithm}[h!]
\caption{Online Score Function Estimation for a given $l\in\mathbb{N}_0$.} \label{alg:online_disc}
\begin{enumerate}
\item{For $i\in\{1,\dots,N\}$, sample $u_{0}^{l,i}$ i.i.d.~from $\eta_{0,\theta}^l(\cdot)$. The estimate of
$\nabla_{\theta}\log(\gamma_{1,\theta}^l(1))$
 is:
%
$$
\widehat{\nabla_{\theta}\log(\gamma_{1,\theta}^l(1))} := \frac{\sum_{i=1}^N G_{0,\theta}^l(u_{-1}^{l,i},u_{0}^{l,i})~\mu_{0,\theta}^l(u_{-1}^{l,i},u_{0}^{l,i})}{\sum_{i=1}^N G_{0,\theta}^l(u_{-1}^{l,i},u_{0}^{l,i})}
$$
with $u_{-1}^{l,i}=x_{*}$ for each $i$.
Set $k=1$, and for $i\in\{1,\dots,N\}$, $\check{u}_{-1}^{l,i}=x_{*}$.}
\item{(Resampling step) For $i\in\{1,\dots,N\}$, sample $\check{u}_{k-1}^{l,i}$ from:
$$
\sum_{i=1}^N\frac{G_{k-1,\theta}^l(\check{u}_{k-2}^{l,i},u_{k-1}^{l,i})}{\sum_{j=1}^N G_{k-1,\theta}^l(\check{u}_{k-2}^{l,j},u_{k-1}^{l,j})}\delta_{\{u_{k-1}^{l,i}\}}(\cdot).
$$
If $k=1$, for $i\in\{1,\dots,N\}$, set $F_{k-1,\theta}^{l,N}(\check{u}_{0}^{l,i})=\mu_{0,\theta}^l(x_{*},\check{u}_{0}^{l,i})$.}
\item{ (Sampling step) For $i\in\{1,\dots,N\}$, sample $u_{k}^{l,i}$ from $M_{\theta}^l(\check{u}_{k-1}^{l,i},\cdot)$. For $i\in\{1,\dots,N\}$, compute:
\begin{equation}
\label{eq:f_rec_def}
F_{k,\theta}^{l,N}(u_{k}^{l,i}) = 
\frac{\sum_{j=1}^N \overline{g}_{k,\theta}^l(\check{x}_k^j,x_{k+\Delta_l}^i)m_{\theta}^l(\check{x}_k^j,x_{k+\Delta_l}^i)\{F_{k-1,\theta}^{l,N}(\check{u}_{k-1}^{l,j})
+ \mu_{k,\theta}^l(\check{u}_{k-1}^{l,j},u_{k}^{l,i})
\}}
{\sum_{j=1}^N \overline{g}_{k,\theta}^l(\check{x}_k^j,x_{k+\Delta_l}^i)m_{\theta}^l(\check{x}_k^j,x_{k+\Delta_l}^i)
}.
\end{equation}
The estimate of
$\nabla_{\theta}\log(\gamma_{k+1,\theta}^l(1))$
 is:
\begin{equation}
\label{eq:grad_est_1}
\widehat{\nabla_{\theta}\log(\gamma_{k+1,\theta}^l(1))} := 
\frac{\sum_{i=1}^N G_{k,\theta}^l(\check{u}_{k-1}^{l,i},u_{k}^{l,i})F_{k,\theta}^{l,N}(u_{k}^{l,i})}{\sum_{i=1}^N G_{k,\theta}^l(\check{u}_{k-1}^{l,i},u_{k}^{l,i})}.
\end{equation}
Set $k=k+1$ and return to the start of 2.}
\end{enumerate}
\end{algorithm}

The cost of computing the estimation of the score function given in \eqref{eq:grad_est_1} is the cost of the particle filter $\mathcal{O}(\Delta^{-1}_l N)$ plus the cost of computing the term $F_{k,\theta}^{l,N}$ which is $\mathcal{O}(N^2)$, thus, the whole cost is $\mathcal{O}(\Delta^{-1}_l N+N^2)$ per unit time. As pointed out in \cite{beskos}, applying multilevel techniques will not improve the complexity of the score function, therefore, we only apply a single-level particle filter. We note that in \cite{beskos} a further extension using diffusion bridges is possible and could be used in the context of this paper; that method has the advantage of a so-called path-space interpretation. One can also consider the methods in \cite{chopin_Hai-Dang}, although it is unclear on its efficacy for the class of models that we consider.


\section{Mathematical Analysis}\label{sec:math}

\begin{hypD}\label{hyp_diff:3}
For any $y\in\mathsf{Y}$ there exists a $0<C<+\infty$ such that $\inf_{x\in\mathbb{R}^d}g(x,y)\geq C$.
\end{hypD}

A sketch proof of the below result can be found in the Appendix. We use the notation $[\pi_t^l-\pi_t^{l-1}](\varphi)=\pi_t^l(\varphi) - \pi_t^{l-1}(\varphi)$.

\begin{prop}\label{prop:coup_main_res1}
Assume (D\ref{hyp_diff:1}-\ref{hyp_diff:3}). Then for any $t\in\mathbb{N}$ there exists a $C<+\infty$ such that for
any $(l,N,\varphi)\in\mathbb{N}^2\times\mathcal{B}_b(\mathbb{R}^{d_x})\cap\textrm{\emph{Lip}}(\mathbb{R}^{d_x})$
$$
\mathbb{E}\left[\Big([\pi_t^l-\pi_t^{l-1}]^N(\varphi)-[\pi_t^l-\pi_t^{l-1}](\varphi)\Big)^2\right] \leq \frac{C(\|\varphi\|+\|\varphi\|_{\textrm{\emph{Lip}}})^2\Delta_l^{1/2}}{N}.
$$
\end{prop}

The following result can be established via Lemma \ref{lem:lem2} and the proofs for Proposition \ref{prop:coup_main_res}. 

\begin{prop}\label{prop:coup_main_res2}
Assume (D\ref{hyp_diff:1}-\ref{hyp_diff:3}). Then for any $t\in\mathbb{N}$ there exists a $C<+\infty$ such that for
any $(l,N,\varphi)\in\mathbb{N}^2\times\mathcal{B}_b(\mathbb{R}^{d_x})\cap\textrm{\emph{Lip}}(\mathbb{R}^{d_x})$
$$
\Big|\mathbb{E}\left[[\pi_t^l-\pi_t^{l-1}]^N(\varphi)-[\pi_t^l-\pi_t^{l-1}](\varphi)\right]\Big| \leq \frac{C(\|\varphi\|+\|\varphi\|_{\textrm{\emph{Lip}}})\Delta_l^{1/4}}{N}.
$$
\end{prop}

Given the above two results, one can show that for the MLPF estimator in \eqref{eq:ml_est}, that the associated MSE (when centering by $\pi_t(\varphi)$) is, for $\epsilon>0$ given, of $\mathcal{O}(\epsilon^2)$ with a computational effort of
$\mathcal{O}(\epsilon^{-2.5})$, with a selection of $N_l=\mathcal{O}(\epsilon^{-2.5}\Delta_l^{3/4})$ with $L$ chosen so that $\mathcal{O}(\Delta_L)=\mathcal{O}(\epsilon)$.
This can be inferred by using standard arguments such as in \cite{mlpf}.

In addition, one has the following result which follows directly from the above results and the proofs in \cite[Theorem 2]{ub_pf}.

\begin{prop}\label{prop:coup_main_res}
Assume (D\ref{hyp_diff:1}-\ref{hyp_diff:3}). Then for any $(t,\varphi)\in\mathbb{N}\cap\mathcal{B}_b(\mathbb{R}^{d_x})\cap\textrm{\emph{Lip}}(\mathbb{R}^{d_x})$ there exist choices of  $\mathbb{P}(p|l)$ and $\mathbb{P}_L(l)$ so that 
\eqref{eq:ub_est} with $\Xi_L$ (resp.~$\Xi_0$) as in \eqref{eq:xi_l} (resp.~\eqref{eq:xi_0}) is an unbiased and finite variance estimator of $\eta_t(\varphi)$.
\end{prop}

We note that the choice of  $\mathbb{P}(p|l)$ and $\mathbb{P}_L(l)$  are exactly as \cite{ub_pf} so we do not discuss it here.  A discussion of the cost is the same as \cite[Section 3.2]{ub_pf} and hence omitted.

\section{Numerical Results}\label{sec:numerics}

\subsection{Models}

In the following, we introduce four different diffusion models (through all the numerical examples in this section and the following we set $d_x=1$). Let $X_0=x^*$ be the initial point at time $t=0$, the diffusion processes that will be associated with the hidden Markov process are given through the following stochastic differential equations (SDEs)

\begin{itemize}
	\item Ornstein-Uhlenbeck (OU):
	Let $\nu,$ $\sigma\in \mathbb{R}^+$. For $t\geq 0$, an Ornstein-Uhlenbeck process is defined as
	\begin{align}
	\label{eq:model1}
	dX_{t}=-\nu X_tdt+\sigma dW_t.
	\end{align}
	\item Langevin process: 
	Let $\pi$ be a probability density function; the overdamped Langevin SDE is defined as 
	\begin{align*}
	dX_{t}=-\nabla \log\left(\pi( X_t)\right)dt+ dW_t.
	\end{align*}
	The asymptotic distribution of this process is precisely $\pi$. In the numerical experiments we use Student's t-distribution with $\nu=10$ degrees of freedom, thus, $\nabla \log(\pi(x))=-(\nu+1)x(x^2+\nu)^{-1}$. Therefore, we have
\begin{align}
\label{eq:model2}
dX_{t}=(\nu+1)X_t(X_t^2+\nu)^{-1} dt+ dW_t.
\end{align}

	\item Nonlinear diffusion term (NLDT): The SDE of this process is defined as 
	\begin{align}
	\label{eq:model3}
	dX_{t}= \frac{1}{\sqrt{1+X_t^2}} dW_t.
	\end{align}
	\item Geometric Brownian Motion (GBM): Let $\nu\in \mathbb{R}$ and $\sigma\in \mathbb{R}^+$, the GBM is defined by the SDE
	\begin{align}
	\label{eq:model4}
	dX_{t}=\nu X_tdt+\sigma X_t dW_t,
	\end{align}
	with $x^*\geq0$.
\end{itemize}

The likelihood $g(x,\cdot)$ is chosen to be the normal distribution with variance $\Sigma\in \mathbb{R}^{+}$ and mean $x \in \mathbb{R}$. The intensity function is taken as $\lambda(x)=a|x|$, where $a\in \mathbb{R}^+$. The observations and the times at which they are observed are generated from running one realization of each of the SDEs above. The parameters in the models above are chosen so that two criteria are met: the first one is the stability of the particle filter for a feasible level of discretization and a certain number of particles, and the second is the strength of the coupling, i.e., how small is the constant $C$ in Propositions \ref{prop:coup_main_res1} compared to the variance of the single-level particle filter. The code of the simulations is written in Python and it can be downloaded from \href{https://github.com/maabs/Multilevel-for-Diffusions-Observed-via-Marked-Point-Processes}{https://github.com/maabs/Multilevel-for-Diffusions-Observed-via-Marked-Point-Processes}.

\subsection{Multilevel Particle Filter}

In \autoref{fig:ML_comple}, we plot the computational cost of the MLPF and compare it with the cost of the single-level PF. We can clearly see the predicted error-to-cost rates of each algorithm and the reduction in cost when using the MLPF compared to using the single-level PF. We choose the final time $t=100$. The number of particles that we need in order to attain the desired MSE $\varepsilon^2$ in the MLPF is $N_l=\mathcal{O}(\Delta_{l}^{3/4}\varepsilon^{-(2.5)})$, for $l\in \{l_0,l_0+1,\cdots,L\}$, where $l_0\in \mathbb{N}_0$. $L=\mathcal{O}(\log(\varepsilon^{-1}))\in \mathbb{N}$, $L > l_0$, depends on the desired level of error.
\begin{figure}[h!]
\centering
\includegraphics[height=0.3\textwidth]{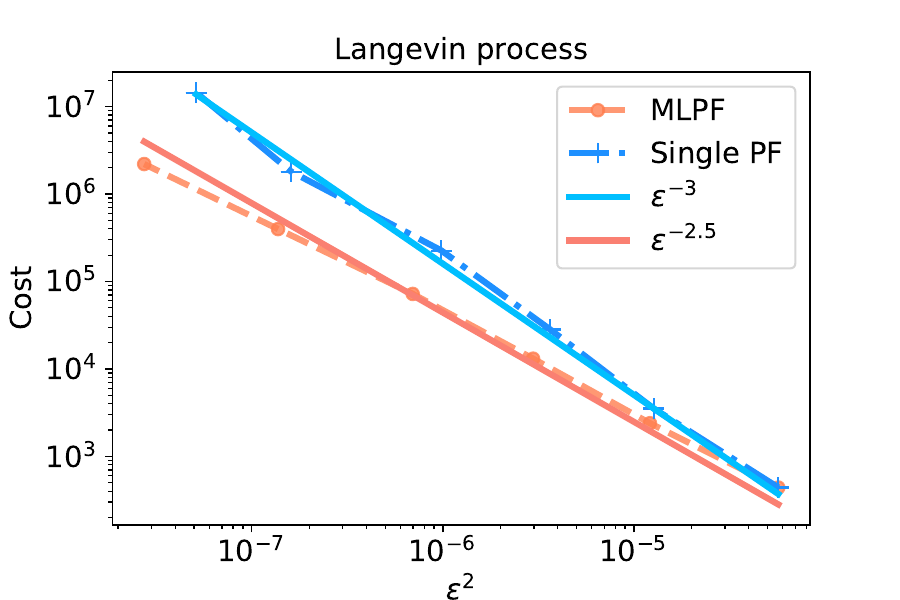}\,
\includegraphics[height=0.3\textwidth]{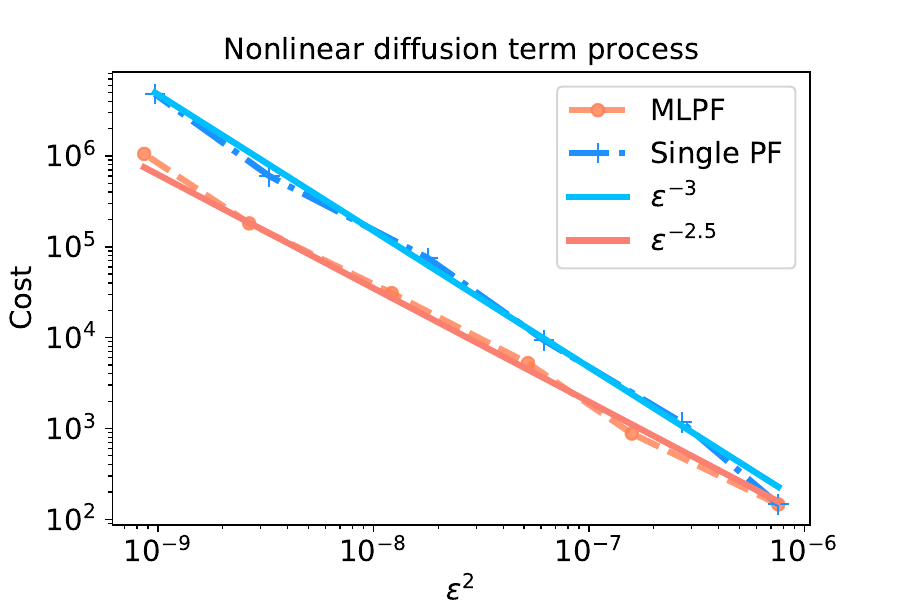}\\
\includegraphics[height=0.3\textwidth]{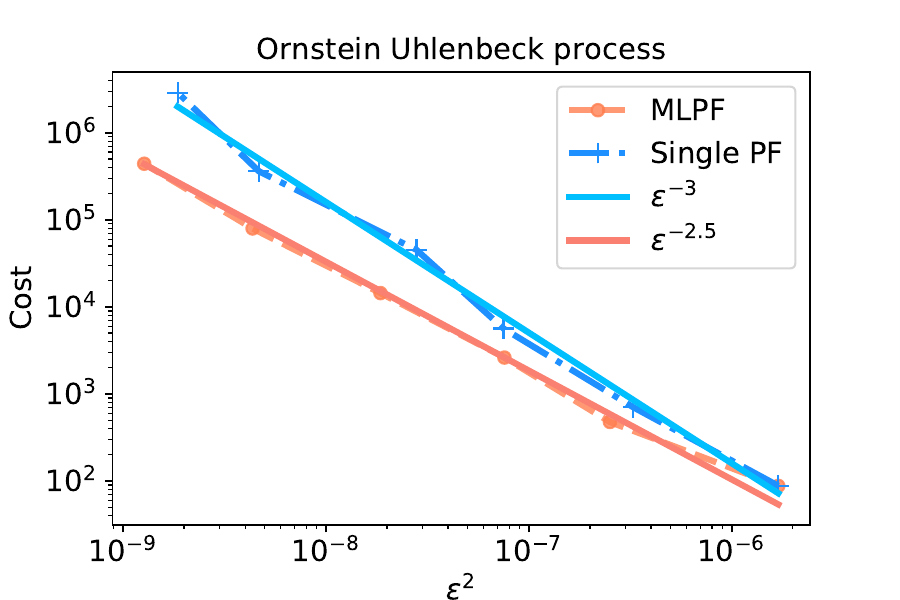}\,
\includegraphics[height=0.3\textwidth]{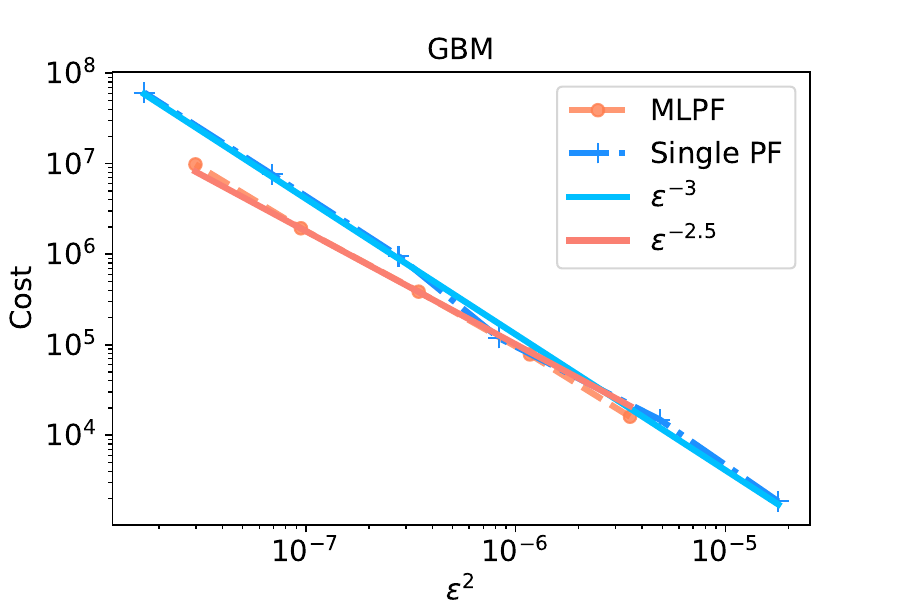}
   \caption{Computional complexity of the multilevel particle filter and the single-level particle filter. On the horizontal line we set MSE $ =\varepsilon^2$.} 
     \label{fig:ML_comple}
\end{figure}

\subsection{Unbiased Particle Filter}
\label{sec:ub_pf_numer}
Similar to the MLPF, we plot the cost vs. MSE of the unbiased estimator in \eqref{eq:ub_est_M} at the final time $t=100$. We use $\mathbb{P}_L(l)\propto \log(l+2)^2(l+1)\Delta^{\frac{1}{2}}_l$ as the probability mass distribution of the discretization level $l$, where $l\in \{l_0,l_0+1,\cdots,L_{\text{Trunc}}\}$. For the number of particles we use a geometric sequence $N_p=N_02^p$, where the randomization of the particles (in terms of $p$) follows the probability mass distribution $\mathbb{P}_P(p)\propto \log(p+2)^2(p+1)/N_p$ for $p\in \{0,1,\cdots, P_{\text{Trunc}} \}$. $L_{\text{Trunc}}$ and $P_{\text{Trunc}}$ are the truncation levels of our now \textit{unbiased} algorithm; we choose them large enough to ensure the bias is negligible w.r.t. the variance of the estimator. For more details on these choices, see \cite{ub_pf}. We use different values of $L_{\text{Trunc}}$, $P_{\text{Trunc}}$ and $N_0$ depending on the diffusion process, for the OU, Langevin and NLDT the values are $L_{\text{Trunc}}=10$, $P_{\text{Trunc}}=11$ and $N_0=5$; for the GBM we have $L_{\text{Trunc}}=10$, $P_{\text{Trunc}}=5$ and $N_0=100$; the value of $l_0=0$ is the same for all the processes.

The computational complexity of the unbiased estimator in \eqref{eq:ub_est_M} is plotted in \autoref{fig:Un_comple}. Due to the random nature of the estimator, we can observe different rates depending on the number of realizations $M$. 
In this Section, we are interested in the rates rather than the comparison to the single-level PF or MLPF; this is because one of the perks of the unbiased estimator comes from its embarrassingly parallel characteristic, thus making the unbiased estimator viable depending on the parallelization capabilities of the hardware. A comparison between the unbiased and multilevel particle filters of partially observed diffusions at regular times can be found in \cite{ub_pf}. 
\begin{figure}[h!]
\centering
\includegraphics[height=0.3\textwidth]{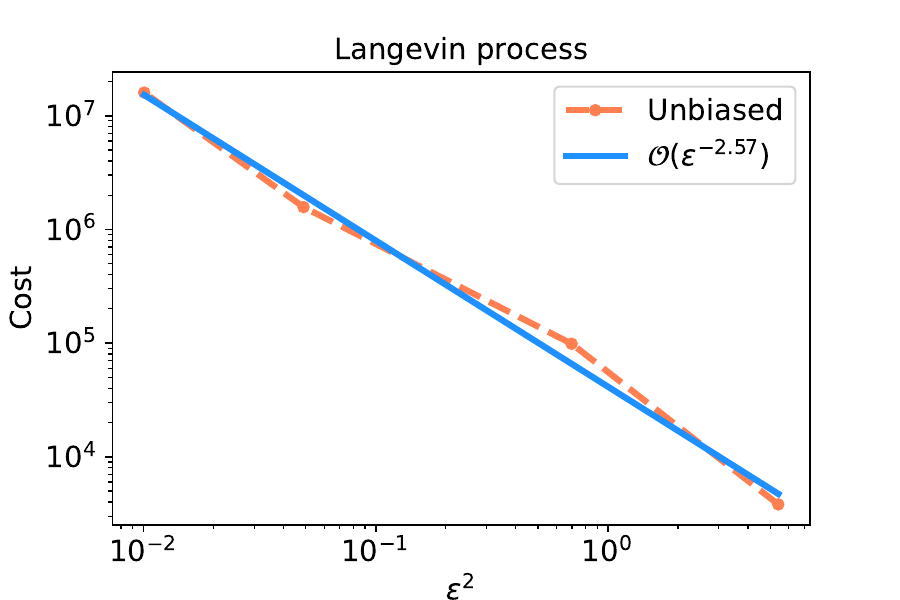}\,
\includegraphics[height=0.3\textwidth]{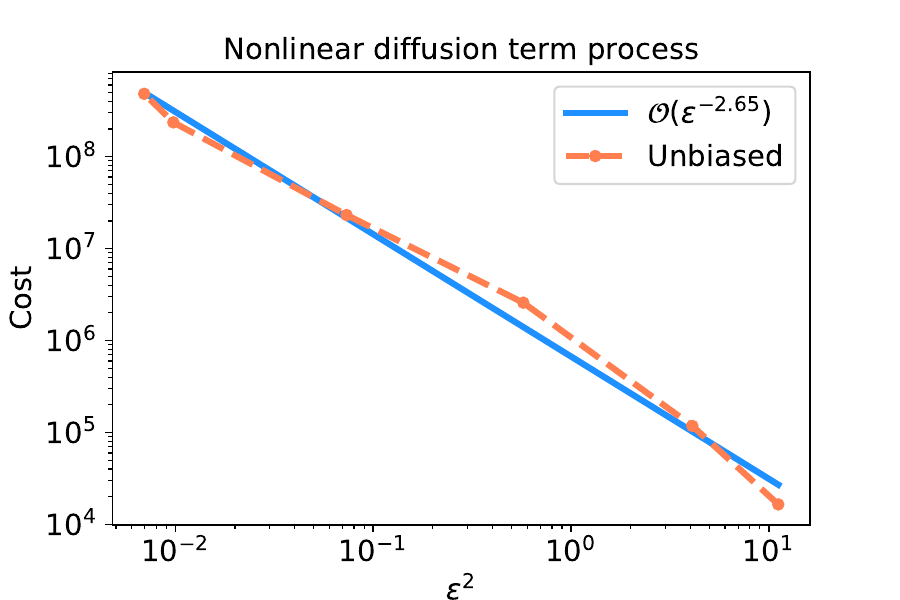}\\
\includegraphics[height=0.3\textwidth]{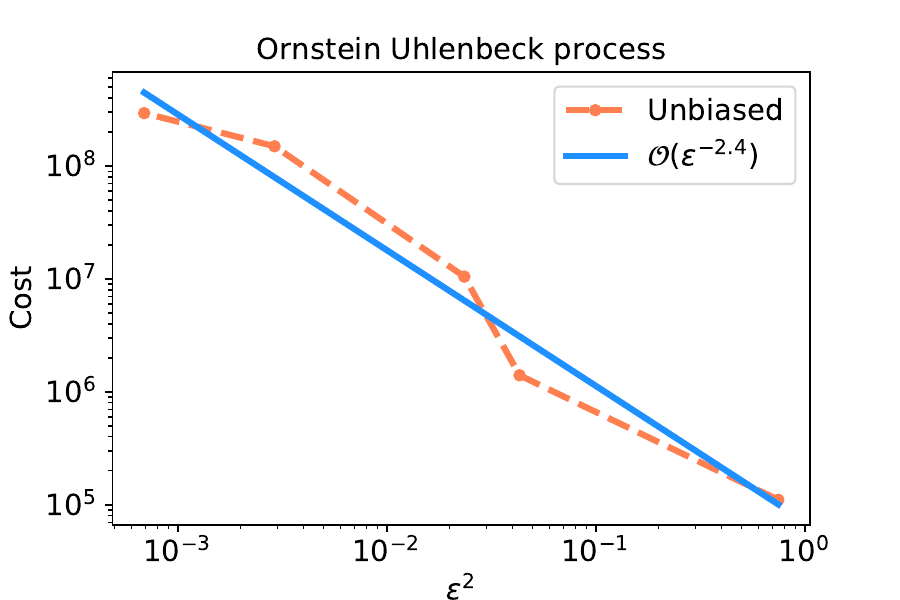}\,
\includegraphics[height=0.3\textwidth]{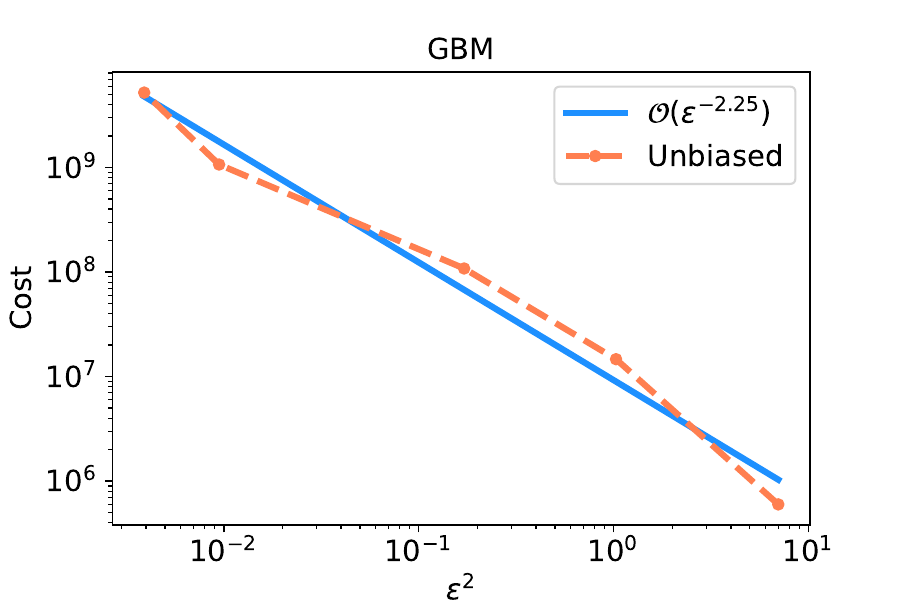}
   \caption{Computational complexity of the unbiased particle filters.} 
     \label{fig:Un_comple}
\end{figure}

\subsection{Parameter Estimation}\label{subsec:par_est}
We use \autoref{alg:online_disc} along with \eqref{eq:sgd} to estimate $\nabla_\theta \log \left(\gamma_{c(k+1), \theta_{0:k}}(1)\right)$ and the parameters $\theta$ for each iteration $k$, specifically, for each $k\geq 0$ we replace $\theta$ by $\theta_k$ in steps 2 and 3 of \autoref{alg:online_disc}. 
In this Section we estimate the parameters for the four different hidden diffusions given above. Also we estimate the parameter $\theta_\lambda\in \mathbb{R}^+$ in the intensity function $\lambda_\theta(x)=\theta_\lambda|x|$, and the and variance parameter $\theta_\Sigma\in \mathbb{R}^+$ in the likelihood function $g_{\theta}(\cdot,y)$, which is defined as a Gaussian density with mean $y$ and variance $\theta_\Sigma$. The hidden processes and their drift terms are: 
\begin{itemize}
\item Ornstein-Uhlenbeck process in \eqref{eq:model1}, with $b_\theta(x)=-\theta_b$ with $\theta_b \in \mathbb{R}^+$. We will estimate $\theta = (\theta_b, \theta_\lambda, \theta_\Sigma)$.
\item Langevin process in \eqref{eq:model2} with $b_{\theta}(x)=(\nu+1)x(\nu+x^2)^{-1}$, $\nu\in \mathbb{R}^+$. We will estimate $\theta = (\theta_\lambda, \theta_\Sigma)$.
\item The nonlinear diffusion process in \eqref{eq:model3} with $b_{\theta}(x)=0$. We will estimate $\theta = (\theta_\lambda, \theta_\Sigma)$.
\item The GBM process in \eqref{eq:model4} with $b_{\theta}(x)=\theta_b x$ with $\theta_b\in \mathbb{R}$. We will estimate $\theta = (\theta_b, \theta_\lambda, \theta_\Sigma)$.
\end{itemize}

In order to control the trade-off between the number of new observations included in the estimation of the score function and the resampling rate of the particle filter, we introduce the hyperparameter $c$ in \eqref{eq:sgd}. Different step-sizes are used depending on the parameter we estimate, i.e., we modify slightly \eqref{eq:sgd} changing $\alpha_m$ from a scalar to a positive definite diagonal matrix in $\mathbb{R}^{d_\theta\times d_\theta}$, with diagonal terms $\alpha_m^{(i)}=\alpha_0^{(i)}(m+1)^{-\beta},$ $\beta \in (0.5,1]$, $\alpha_0^{(i)}\in \mathbb{R}^+$, $i\in\{1,2,3\}$, $m \in \mathbb{N}_0$. Plots of the evolving values of $\theta_m$ are provided in \autoref{fig:sgd}. For the OU process we choose the true parameters to be $(\overline{\theta}_\lambda,\overline{\theta}_\Sigma,\overline{\theta}_b)=(3.5,1,0.98)$ and the initial guesses $({\theta}_{\lambda,0},{\theta}_{\Sigma,0},{\theta}_{b,0})=(1.5,1.5,0.48)$. For the Langevin process the true parameters are $(\overline{\theta}_\lambda,\overline{\theta}_\Sigma)=(1,1)$ and the initial values $({\theta}_{\lambda,0},{\theta}_{\Sigma,0})=(2,2.5)$. For the nonlinear diffusion process the true parameters are $(\overline{\theta}_\lambda,\overline{\theta}_\Sigma)=(0.222,1)$ and the initial values $({\theta}_{\lambda,0},{\theta}_{\Sigma,0})=(2.222,2)$. Finally, for the GBM the true parameters are $(\overline{\theta}_\lambda,\overline{\theta}_\Sigma,\overline{\theta}_b)=(0.5,1,0.015)$ and $({\theta}_{\lambda,0},{\theta}_{\Sigma,0},{\theta}_{b,0})=(2.5,2,1.015)$. We can see in the figures the relatively fast convergence to the true parameters in all models considered.

\begin{figure}[h!]
\centering
\includegraphics[height=0.3\textwidth]{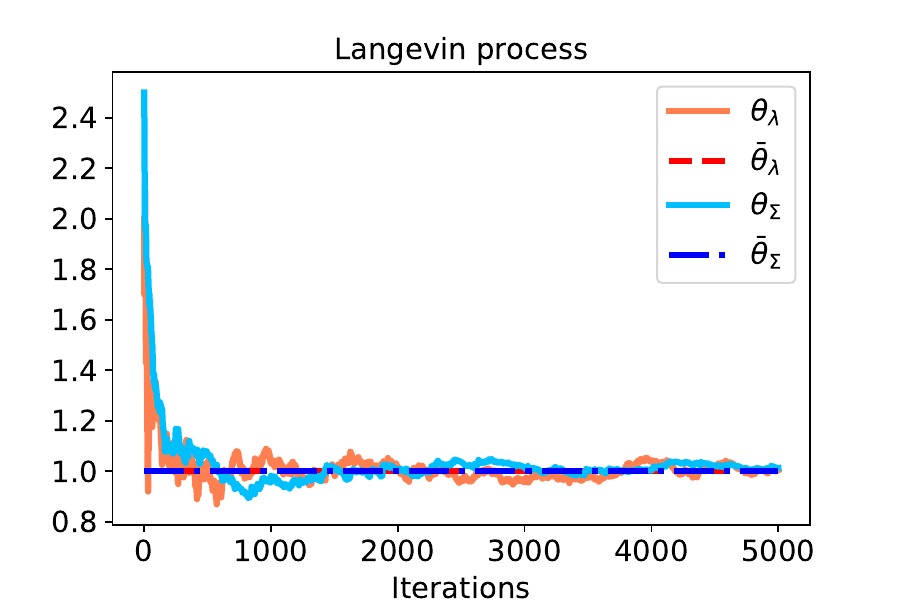}\,
\includegraphics[height=0.3\textwidth]{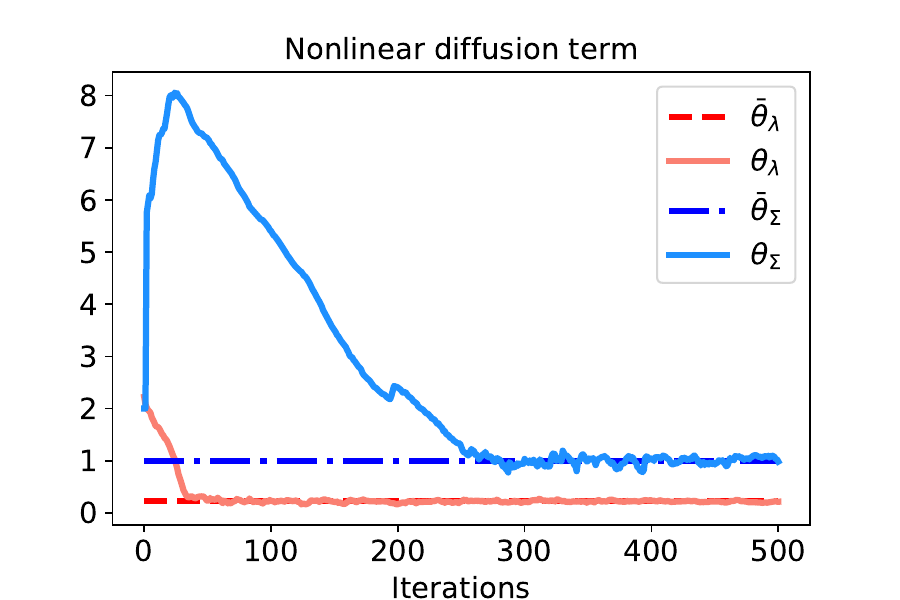}\\
\includegraphics[height=0.3\textwidth]{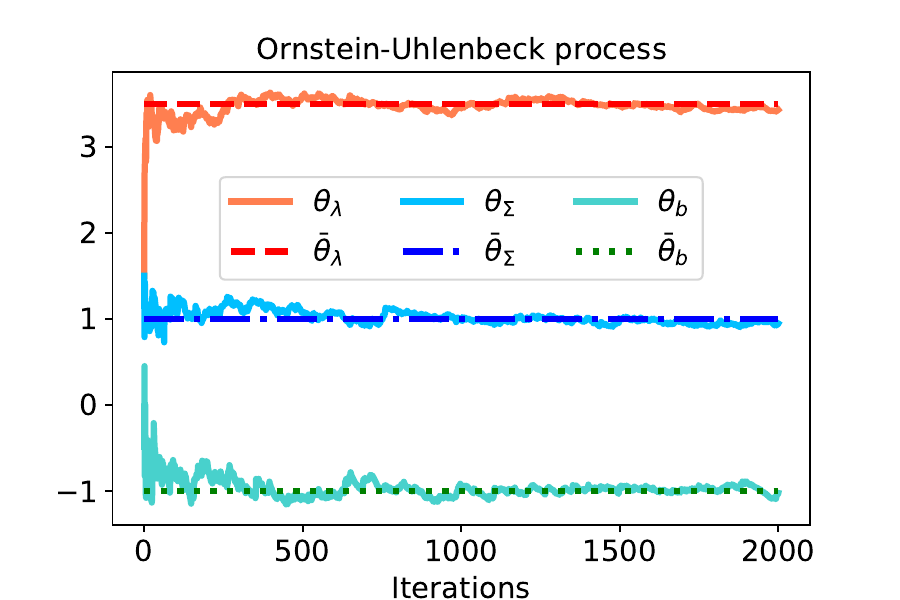}\,
\includegraphics[height=0.3\textwidth]{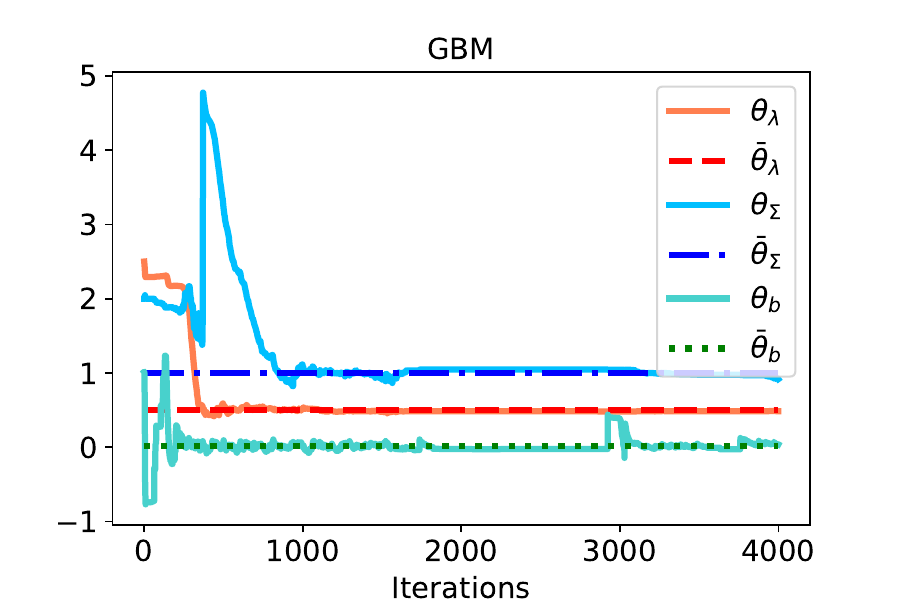}
   \caption{Estimated parameters in terms of the iterations for the four different processes.} 
     \label{fig:sgd}
\end{figure}



\subsubsection*{Acknowledgements}

All authors were supported by KAUST baseline funding.

\appendix

\section{Mathematical Proofs}

The proof of Proposition \ref{prop:coup_main_res} is virtually identical to that of \cite[Theorem 4.1]{jasra}.
The only issue is that there is minor mistake in the proof of \cite[Lemma A.1.]{jasra} and we correct that here. As the rest  of the proof is then basically as \cite[Theorem 4.1]{jasra} we do not repeat the technical proofs. The reason that the proofs are so similar, is that our hidden process is identical to that of \cite{jasra} and the potential functions $\mathbf{G}_t^l$ (defined below) are uniformly in $l$ upper and lower bounded by deterministic constants as well having similar functional form to the counter-parts in \cite{jasra}. To that end, we give a similar (and corrected) result to \cite[Lemma A.1.]{jasra} and leave the remainder of the proof as a simple exercise in reading and adapting the results in \cite{jasra}.

\subsection{Notations}

To assist moving between the proofs of this paper and that of \cite{jasra} we adopt a similar notations to that paper, which are different from the main text.

For $p\in\mathbb{N}_0$ set
$$
u_p^l := (x_{p},x_{p+\Delta_l},\dots,x_{p+1}) \in (\mathbb{R}^{d_x})^{\Delta_l^{-1}+1} =: E_l
$$
For technical reasons, this notation differs slightly from that in Section \ref{sec:disc_filter}.
  
For $\varphi\in\mathcal{B}_b(\mathbb{R}^{d_x})$, we define, for any $l\geq 0$, $\pmb{\varphi}^l:E_l\rightarrow\mathbb{R}$
$$
\pmb{\varphi}^l(x_0,x_{\Delta_l},\dots,x_{1}) := \varphi(x_1).
$$
Set, with $p\in\mathbb{N}_0$
$$
\mathbf{G}_p^l(u_p^l) := \left(\prod_{k=n_{p}+1}^{n_{p+1} }g(x_{s_k},y_{s_k})\lambda(x_{s_k})\right)\exp\left\{-\Delta_l\sum_{k=0}^{\Delta_l^{-1}-1}\lambda(x_{p+(k+1)\Delta_l})\right\}.
$$
Denote by $M^l:\mathbb{R}^{d_x}\rightarrow\mathcal{P}(E_l)$ the joint Markov transition of $(x_0,x_{\Delta_l},\dots,x_{1})$ defined via the Euler discretization and a Dirac on a point $x\in\mathbb{R}^{d_x}$:  for $(x,\varphi)\in\mathbb{R}^{d_x}\times\mathcal{B}_b(E_l)$, 
$$
M^l(\varphi)(x) := \int_{E_l}\varphi(x_0,x_{\Delta_l},\dots,x_{1})\delta_x(dx_0)\Big[\prod_{k=1}^{\Delta_l^{-1}}m(x_{(k-1)\Delta_l},x_{k\Delta_l})\Big]d(x_{\Delta_l},\dots,x_{1}).
$$
For $p\in\mathbb{N}$, we use the notation 
$$
\mu(\mathbf{G}_{p-1}^l\mathbf{M}^l(\varphi))= \int_{E_l}\mu(d(x_{p-1},x_{p-1+\Delta_l},\dots,x_{p}))\mathbf{G}_{p-1}^l(x_{p-1},x_{p-1+\Delta_l},\dots,x_{p-\Delta_l}, x_p)M^l(\varphi)(x_p)
$$ 
where $(\mu,\varphi)\in\mathcal{P}(E_l)\times\mathcal{B}_b(E_l)$. 
For $p\in\mathbb{N}$, define the operator $\Phi_p^l:\mathcal{P}(E_l)\rightarrow\mathcal{P}(E_l)$ with $(\mu,\varphi)\in\mathcal{P}(E_l)\times\mathcal{B}_b(E_l)$ as:
$$
\Phi_p^l(\mu)(\varphi) :=  \frac{\mu(\mathbf{G}_{p-1}^l\mathbf{M}^l(\varphi))}{\mu(\mathbf{G}_{p-1}^l)}.
$$

Now, define, for $(p,\varphi)\in\mathbb{N}_0\times\mathcal{B}_b(E_l)$, 
$$
\eta_p^l(\varphi) := \int_{\mathbb{R}^{d_x}}\eta_{p}^l(dx) \int_{E_l} M^l(x,du)\varphi(u).
$$
Then one can establish that for $p\in\mathbb{N}$
$$
\eta_p^l(\varphi) = \Phi_p^l(\eta_{p-1}^l)(\varphi).
$$
Moreover, for $(t,\varphi)\in\mathbb{N}\times\mathcal{B}_b(\mathbb{R}^{d_x})$
$$
\pi_t^l(\varphi) = \frac{\eta_{t-1}^l(\mathbf{G}_{t-1}^l\pmb{\varphi}^l)}{\eta_{t-1}^l(\mathbf{G}_{t-1}^l)}
$$
which is the time discretized filter.

Some operators are now defined. Let $(l,p,n)\in\mathbb{N}_0^3$, $n>p$, $(x,\varphi)\in E_l\times\mathcal{B}_b(E_l)$
$$
\mathbf{Q}_{p,n}^l(\varphi)(u_p) := \int \varphi(u_n)\Big(\prod_{q=p}^{n-1} \mathbf{G}_q^l(u_q)\Big) \prod_{q=p+1}^{n}M^l(u_{q-1},du_q).
$$
where we use the convention $\mathbf{Q}_{p,p}^l(\varphi)(u_p)=\varphi(u_p)$. 
In addition, for $(l,p,n)\in\mathbb{N}_0^3$, $n>p$, $(x,\varphi)\in E_l\times\mathcal{B}_b(E_l)$:
$$
\mathbf{D}_{p,n}^l(\varphi)(u_p) := \frac{\mathbf{Q}_{p,n}^l(\varphi-\eta_n^l(\varphi))(u_p)}{\pi_p^l(\mathbf{Q}_{p,n}^l(1))}
$$
where $\mathbf{D}_{p,p}^l(\varphi)(u_p)=\varphi(u_p)-\eta_p^l(\varphi)$.
Throughout our arguments, $C$ is a finite constant whose value may change from line to line, but does not depend upon $l$ nor $N$. The particular dependencies of a given constant will be clear from the statement of a given result. 

Set, for $l\in\mathbb{N}$, $(n,p,\varphi)\in\mathbb{N}_0^2\times\mathcal{B}_b(\mathbb{R}^{d_x})$, $p< n$
\begin{align*}
T_{p,n}^{l}(\varphi)  :=  \mathbb{E}[(\mathbf{D}_{p,n}^l(\mathbf{G}_n^l\pmb{\varphi}^l)(U_p^{l,1})- \mathbf{D}_{p,n}^{l-1}(\mathbf{G}_n^{l-1}\pmb{\varphi}^{l-1})(\bar{U}_p^{l-1,1}))^4]^{1/2} + 
\|\varphi\|^2\mathbb{E}[(\mathbf{G}_{p}^l(U_p^{l,1})-\mathbf{G}_{p}^{l-1}(\bar{U}_p^{l-1,1}))^4]^{1/2} + \Delta_l^2
\end{align*}
and if $p=n$
$$
T_{p,n}^{l}(\varphi) :=  \mathbb{E}[(\mathbf{D}_{p,n}^l(\mathbf{G}_n^l\pmb{\varphi}^l)(U_p^{l,1})- \mathbf{D}_{p,n}^{l-1}(\mathbf{G}_n^{l-1}\pmb{\varphi}^{l-1})(\bar{U}_p^{l-1,1}))^4]^{1/2} + \Delta_l^2.
$$

We remark that using these modified operators in the analysis of the multilevel and unbiased methods do not change the algorithms and are just an alternative representation. In the notation just introduced, we use $\check{\eta}_t^{l,N}$ and $\check{\bar{\eta}}_t^{l-1,N}$ to denote the $N-$empirical measures from the coupled particle filter at levels $l$ and $l-1$ associated to the samples $u_{n}^{l,1:N}$ and $\bar{u}_{n}^{l-1,1:N}$.

\subsection{Technical Results}

\begin{lem}\label{lem:lem1}
Assume (D\ref{hyp_diff:1}-\ref{hyp_diff:3}). Then for any $n\in\mathbb{N}$ there exists a $C<+\infty$ such that for
any $(l,N,\varphi)\in\mathbb{N}^2\times\mathcal{B}_b(\mathbb{R}^{d_x})$
$$
\mathbb{E}\left[\left(\check{\eta}_{n}^{l,N}(\mathbf{G}_{n}^l\pmb{\varphi}^l)-\check{\bar{\eta}}_{n}^{l-1,N}(\mathbf{G}_{n}^{l-1}\pmb{\varphi}^{l-1})
- \{\eta_{n}^{l}(\mathbf{G}_{n}^l\pmb{\varphi}^l) - \eta_{n}^{l-1}(\mathbf{G}_{n}^{l-1}\pmb{\varphi}^{l-1})
\}
\right)^2\right] \leq \frac{C}{N}\sum_{p=0}^n T_{p,n}^l(\varphi).
$$
\end{lem}

\begin{proof}
As in  \cite[Lemma A.1.]{jasra} we apply the following standard Martingale plus remainder decomposition \cite[Lemma 6.3]{ddj2012} followed by the $C_2-$inequality multiple times to yield the upper-bound
$$
\mathbb{E}\left[\left(\check{\eta}_{n}^{l,N}(\mathbf{G}_{n}^l\pmb{\varphi}^l)-\check{\bar{\eta}}_{n}^{l-1,N}(\mathbf{G}_{n}^{l-1}\pmb{\varphi}^{l-1})
- \{\eta_{n}^{l}(\mathbf{G}_{n}^l\pmb{\varphi}^l) - \eta_{n}^{l-1}(\mathbf{G}_{n}^{l-1}\pmb{\varphi}^{l-1})
\}
\right)^2\right]  \leq C\Big(\sum_{p=0}^n\mathbb{E}[T_1(p)^2] + \sum_{p=0}^{n-1}\mathbb{E}[T_2(p)^2]\Big)
$$
where
\begin{eqnarray*}
T_1(p) & := & 
(\check{\eta}_p^{l,N}-\Phi_p^l(\check{\eta}_{p-1}^{l,N}))(\mathbf{D}_{p,n}^l(\mathbf{G}_n^l\pmb{\varphi}^l)) - 
(\check{\bar{\eta}}_p^{l-1,N}-\Phi_p^{l-1}(\check{\bar{\eta}}_{p-1}^{l-1,N}))(\mathbf{D}_{p,n}^{l-1}(\mathbf{G}_n^{l-1}\pmb{\varphi}^{l-1})) \\
T_2(p) & := & 
\left(\frac{\check{\eta}_p^{l,N}(\mathbf{D}_{p,n}^l(\mathbf{G}_n^l\pmb{\varphi}^l))}{\check{\eta}_p^{l,N}(\mathbf{G}_p^l)}[\eta_p^l-\check{\eta}_p^{l,N}](\mathbf{G}_p^l) - 
\frac{\check{\bar{\eta}}_p^{l-1,N}(\mathbf{D}_{p,n}^{l-1}(\mathbf{G}_n^{l-1}\pmb{\varphi}^{l-1}))}{\check{\bar{\eta}}_p^{l-1,N}(\mathbf{G}_p^{l-1})}[\eta_p^{l-1}-\check{\bar{\eta}}_p^{l-1,N}](\mathbf{G}_p^{l-1})\right).
\end{eqnarray*}
We have to control the terms $T_1(p)$, $p\in\{0,1,\dots,n\}$ and  $T_2(p)$, $p\in\{0,1,\dots,n-1\}$ in an appropriate way. 

For $T_1(p)$ applying the conditional Marcinkiewicz-Zygmund inequality followed by Jensen's inequality
\begin{equation}\label{eq:t1p}
\mathbb{E}[T_1(p)^2] \leq \frac{C}{N}\overline{\mathbb{E}}[
(\mathbf{D}_{p,n}^l(\mathbf{G}_n^l\pmb{\varphi}^l)(U_p^{l,1})-\mathbf{D}_{p,n}^{l-1}(\mathbf{G}_n^{l-1}\pmb{\varphi}^{l-1})(\bar{U}_p^{l-1,1}))^4]^{1/2}.
\end{equation}
For $T_2(p)$ we have
$$
T_2(p) = T_3(p) + T_4(p) + T_5(p),
$$
where
\begin{eqnarray}
T_3(p) & := & [\eta_p^l-\check{\eta}_p^{l,N}](\mathbf{G}_p^l)\frac{\check{\eta}_p^{l,N}(\mathbf{D}_{p,n}^l(\mathbf{G}_n^l\pmb{\varphi}^l))}{\check{\eta}_p^{l,N}(\mathbf{G}_p^l)\check{\bar{\eta}}_p^{l-1,N}(\mathbf{G}_p^{l-1})}\Big\{\check{\bar{\eta}}_p^{l-1,N}(\mathbf{G}_p^{l-1})-\check{\eta}_p^{l,N}(\mathbf{G}_p^l)\Big\}\label{eq:t3p}\\
T_4(p) & := & [\eta_p^l-\check{\eta}_p^{l,N}](\mathbf{G}_p^l)\frac{1}{\check{\bar{\eta}}_p^{l-1,N}(\mathbf{G}_p^{l-1})}\Big\{\check{\eta}_p^{l,N}(\mathbf{D}_{p,n}^l(\mathbf{G}_n^l\pmb{\varphi}^l))-\check{\bar{\eta}}_p^{l-1,N}(\mathbf{D}_{p,n}^{l-1}(\mathbf{G}_n^{l-1}\pmb{\varphi}^{l-1}))\Big\}\label{eq:t4p}\\
T_5(p) & := & \frac{\check{\bar{\eta}}_p^{l-1,N}(\mathbf{D}_{p,n}^{l-1}(\mathbf{G}_n^{l-1}\pmb{\varphi}^{l-1}))}{\check{\bar{\eta}}_p^{l-1,N}(\mathbf{G}_p^{l-1})}
\left(
[\eta_p^l-\check{\eta}_p^{l,N}](\mathbf{G}_p^l)
-[\eta_p^{l-1}-\check{\bar{\eta}}_p^{l-1,N}](\mathbf{G}_p^{l-1})
\right).\label{eq:t5p}
\end{eqnarray}
Using the $C_2-$inequality, we need only to bound the second moment of each of the terms $T_3(p)$, $T_4(p)$ and $T_5(p)$ to conclude the proof.

For $T_3(p)$ applying Cauchy-Schwarz twice and using the uniform in $l$ lower-bounds on $\mathbf{G}_p^l$ and $\mathbf{G}_p^{l-1}$ yields that
$$
\mathbb{E}[T_3(p)^2] \leq C\mathbb{E}[|[\eta_p^l-\check{\eta}_p^{l,N}](\mathbf{G}_p^l)|^8]^{1/4}
\mathbb{E}[|\check{\eta}_p^{l,N}(\mathbf{D}_{p,n}^l(\mathbf{G}_n^l\pmb{\varphi}^l))|^8]^{1/4}
\mathbb{E}[|\check{\bar{\eta}}_p^{l-1,N}(\mathbf{G}_p^{l-1})-\check{\eta}_p^{l,N}(\mathbf{G}_p^l)|^{4}]^{1/2}.
$$
The first two terms on the R.H.S.~are, uniformly in $l$, $\mathcal{O}(N^{-1})$ (see e.g.~the proofs in \cite[Appendix A.5]{jasra}) and the last term is upper-bounded by $\mathbb{E}[(\mathbf{G}_{p}^l(U_p^{l,1})-\mathbf{G}_{p}^{l-1}(\bar{U}_p^{l-1,1}))^4]^{1/2}$, which gives
$$
\mathbb{E}[T_3(p)^2] \leq C\frac{\|\varphi\|^2}{N^2}\mathbb{E}[(\mathbf{G}_{p}^l(U_p^{l,1})-\mathbf{G}_{p}^{l-1}(\bar{U}_p^{l-1,1}))^4]^{1/2}.
$$
For $T_4(p)$ applying Cauchy-Schwarz and using the uniform in $l$ lower-bound on $G_p^{l-1}$ gives
$$
\mathbb{E}[T_4(p)^2] \leq C\mathbb{E}[|[\eta_p^l-\check{\eta}_p^{l,N}](\mathbf{G}_p^l)|^4]^{1/2}
\mathbb{E}[|\check{\eta}_p^{l,N}(\mathbf{D}_{p,n}^l(\mathbf{G}_n^l\pmb{\varphi}^l))-\check{\bar{\eta}}_p^{l-1,N}(\mathbf{D}_{p,n}^{l-1}(\mathbf{G}_n^{l-1}\pmb{\varphi}^{l-1}))|^4]^{1/2},
$$
thus, it easily follows that
$$
\mathbb{E}[T_4(p)^2] \leq \frac{C}{N}\mathbb{E}[(\mathbf{D}_{p,n}^l(\mathbf{G}_n^l\pmb{\varphi}^l)(U_p^{l,1})- \mathbf{D}_{p,n}^{l-1}(\mathbf{G}_n^{l-1}\pmb{\varphi}^{l-1})(\bar{U}_p^{l-1,1}))^4]^{1/2}.
$$
One can use similar arguments to the above to establish that
$$
\mathbb{E}[T_5(p)^2] \leq C\frac{\|\varphi\|^2}{N}\left(\mathbb{E}[(\mathbf{G}_{p}^l(U_p^{l,1})-\mathbf{G}_{p}^{l-1}(\bar{U}_p^{l-1,1}))^4]^{1/2} + \Delta_l^2\right)
$$
where the $\Delta_l^2$ term is from the weak error of the filter/predictor (see Proposition \ref{prop:weak_error}).
Therefore, we have shown that
\begin{equation}\label{eq:t2p_bd}
\mathbb{E}[T_2(p)^2] \leq C T_{p,n}^l(\varphi).
\end{equation}
Combining \eqref{eq:t1p} and \eqref{eq:t2p_bd} the proof can easily be concluded.
\end{proof}

\begin{rem}\label{rem:rem1}
Under (D\ref{hyp_diff:1}-\ref{hyp_diff:3}) one can also prove that: for any $(n,q)\in\mathbb{N}^2$, $n\leq q$ there exists a $C<+\infty$ such that for
any $(l,N,\varphi)\in\mathbb{N}^2\times\mathcal{B}_b(\mathbb{R}^{d_x})$
{\small $$
\mathbb{E}\left[\left(\check{\eta}_{n}^{l,N}(\mathbf{D}_{n,q}^l(\mathbf{G}_{q}^l\pmb{\varphi}^l))-\check{\bar{\eta}}_{n}^{l-1,N}(\mathbf{D}_{n,q}^{l-1}(\mathbf{G}_{q}^{l-1}\pmb{\varphi}^{l-1}))
- \{\eta_{n}^{l}(\mathbf{D}_{n,q}^l(\mathbf{G}_{q}^l\pmb{\varphi}^l)) - \eta_{n}^{l-1}(\mathbf{D}_{n,q}^{l-1}(\mathbf{G}_{q}^{l-1}\pmb{\varphi}^{l-1}))
\}
\right)^2\right] \leq \frac{C}{N}\sum_{p=0}^n T_{p,n,q}^l(\varphi),
$$
}
where  for $p< n$
\begin{align*}
T_{p,n,q}^{l}(\varphi) :=  \mathbb{E}[(\mathbf{D}_{p,q}^l(\mathbf{G}_q^l\pmb{\varphi}^l)(U_p^{l,1})- \mathbf{D}_{p,q}^{l-1}(\mathbf{G}_q^{l-1}\pmb{\varphi}^{l-1})(\bar{U}_p^{l-1,1}))^4]^{1/2} + 
\|\varphi\|^2\mathbb{E}[(\mathbf{G}_{p}^l(U_p^{l,1})-\mathbf{G}_{p}^{l-1}(\bar{U}_p^{l-1,1}))^4]^{1/2} + \Delta_l^2
\end{align*}
and if $p=n$
$$
T_{p,n,q}^{l}(\varphi) :=  \mathbb{E}[(\mathbf{D}_{p,q}^l(\mathbf{G}_q^l\pmb{\varphi}^l)(U_p^{l,1})- \mathbf{D}_{p,q}^{l-1}(\mathbf{G}_q^{l-1}\pmb{\varphi}^{l-1})(\bar{U}_p^{l-1,1}))^4]^{1/2} + \Delta_l^2.
$$
\end{rem}

\begin{lem}\label{lem:lem2} 
Assume (D\ref{hyp_diff:1}-\ref{hyp_diff:3}). Then for any $n\in\mathbb{N}$ there exists a $C<+\infty$ such that for
any $(l,N,\varphi)\in\mathbb{N}^2\times\mathcal{B}_b(\mathbb{R}^{d_x})$
$$
\left|\mathbb{E}\left[\check{\eta}_{n}^{l,N}(\mathbf{G}_{n}^l\pmb{\varphi}^l)-\check{\bar{\eta}}_{n}^{l-1,N}(\mathbf{G}_{n}^{l-1}\pmb{\varphi}^{l-1})
- \{\eta_{n}^{l}(\mathbf{G}_{n}^l\pmb{\varphi}^l) - \eta_{n}^{l-1}(\mathbf{G}_{n}^{l-1}\pmb{\varphi}^{l-1})
\}
\right]\right| \leq 
$$
$$
\frac{C}{N}\sum_{p=0}^n\left[\left(\sum_{q=0}^p T_{q,p}^l(\varphi)\right)^{1/2} + 
\left(\sum_{p=0}^n T_{p,n,q}^l(\varphi)\right)^{1/2}
\right].
$$
\end{lem}

\begin{proof}
Using the Martingale plus remainder decomposition we have that
$$
\left|\mathbb{E}\left[\check{\eta}_{n}^{l,N}(\mathbf{G}_{n}^l\pmb{\varphi}^l)-\check{\bar{\eta}}_{n}^{l-1,N}(\mathbf{G}_{n}^{l-1}\pmb{\varphi}^{l-1})
- \{\eta_{n}^{l}(\mathbf{G}_{n}^l\pmb{\varphi}^l) - \eta_{n}^{l-1}(\mathbf{G}_{n}^{l-1}\pmb{\varphi}^{l-1})
\}
\right]\right| \leq 
\sum_{j=3}^5 \mathbb{E}[T_j(p)]
$$
where the terms $T_{3:5}(p)$ are given in \eqref{eq:t3p}-\eqref{eq:t5p}. The proof essentially follows that of Lemma \ref{lem:lem1} except that one will use Lemma \ref{lem:lem1}  and Remark \ref{rem:rem1} where it is appropriate; as the proof is fairly trivial, therefore, it is omitted for brevity.
\end{proof}

\end{document}